\documentclass[final,a4paper,UKenglish,thm-restate,autoref,numberwithinsect]{lipics-v2021}
\pdfoutput=1
\nolinenumbers
\usepackage[utf8]{inputenc}
\usepackage[T1]{fontenc}

\usepackage{tikz}
\usepackage{caption}
\usepackage{arydshln}
\usepackage{array, multirow, bigdelim, makecell, booktabs} 
\usepackage{xcolor,colortbl}

\definecolor{LightGray}{rgb}{0.88,1,1}

\usepackage[justification=centering]{caption}
\newcolumntype{C}[1]{>{\centering\arraybackslash}p{#1}}

\hypersetup{
    colorlinks=true,
    linkcolor=black,
    citecolor=black,
    filecolor=black,
    urlcolor=black,
    pdftitle={Syntactic Minimization of Nondeterministic Finite Automata},
    pdfauthor={Robert S.~R. Myers, Henning Urbat},
    pdfkeywords={Regular languages, Nondeterministic automata, Algebraic language theory},
    pdfduplex={DuplexFlipLongEdge},
}
  
\usepackage{lipsum}
\usepackage{mathtools}

\newcounter{mycounter}

\usepackage{amsmath,mathrsfs,amsxtra,MnSymbol,stmaryrd}
\usepackage[bbgreekl]{mathbbol}
\usepackage{xspace}
\usepackage{dsfont}
\usepackage{chngcntr}
\usepackage{float}
\usepackage[all]{xy}
\SelectTips {cm}{}
\newdir{ >}{{}*!/-5pt/@{>}}
\newdir{ (}{{}*!/-5pt/@{(}}

\DeclareFontFamily{U}{bbold}{}\DeclareFontShape{U}{bbold}{m}{n}{%
  <4.25>bbold5<5>bbold5<6>bbold6<7>bbold7<8>bbold8<8.5>bbold9<9.5>bbold10<10>bbold10<12>bbold12}{}

\usepackage[notcite,notref,final]{showkeys}

\usepackage{seqsplit}
\usepackage{xstring}
\usepackage{xcolor}

\makeatletter
\renewcommand*\showkeyslabelformat[1]{%
\@ifundefined{hideNextShowKeysLabel}{%
\noexpandarg%
\StrSubstitute{#1}{ }{\textvisiblespace}[\TEMP]%
\parbox[t]{\marginparwidth}{\raggedright\normalfont\small\ttfamily\(\{\){\color{red!50!black}\expandafter\seqsplit\expandafter{\TEMP}}\(\}\)}%
}{}
}
\makeatother

\theoremstyle{definition}
\newtheorem{defn}[theorem]{Definition}

\newtheorem{rem}[theorem]{Remark}
\newtheorem{notation}[theorem]{Notation}

\newtheorem*{theorem*}{Theorem}

\usepackage[inline]{enumitem}
\setlist[enumerate,1]{label=(\arabic*),font=\normalfont,align=left,leftmargin=0pt,labelindent=0pt,listparindent=\parindent,labelwidth=0pt,itemindent=!,topsep=0pt,parsep=0pt,itemsep=0pt,start=1}
\setlist[enumerate,2]{label=(\alph*),font=\normalfont,labelindent=*,topsep=0pt,leftmargin=*,start=1}
\setlist[itemize]{labelindent=*,leftmargin=*,topsep=5pt,itemsep=3pt}
\setlist[description]{labelindent=*,leftmargin=*,itemindent=-1 em}

\def\C{\mathcal{C}}

\def\epsilon{\varepsilon}

\renewcommand{\rho}{\varrho}
\def\ol{\overline}

\newcommand{\takeout}[1]{\empty}

\newcommand{\rep}{\mathsf{rep}}
\newcommand{\red}{\mathsf{red}}

\renewcommand{\phi}{\varphi}

\newcommand{\xra}{\xrightarrow}

\newcommand{\Lra}{\Leftrightarrow}
\newcommand{\seq}{\subseteq}

\renewcommand{\intop}{\mathsf{in}}

\newcommand{\id}{\mathsf{id}}
\newcommand{\Id}{\mathsf{Id}}

\newcommand{\NP}{\mathsf{NP}}
\newcommand{\DFA}{\mathbf{DFA}}
\newcommand{\NFAmu}{\mathbf{NFA}_\mathbf{syn}}

\newcommand{\MONTOSUBATOMICNFA}{\mathbf{MON\to NFA_\mathbf{syn}}}
\newcommand{\BICCOV}{\mathbf{BICLIQUE~COVER}}
\newcommand{\PSPACE}{\mathsf{PSPACE}}

\newcommand{\JSL}{{\mathbf{JSL}}}
\newcommand{\JSLf}{{\mathbf{JSL_{\mathbf{f}}}}}

\renewcommand{\dim}[1]{\mathsf{dim}(#1)}

\newcommand{\Rel}{\mathbf{Rel}}

\newcommand{\under}[1]{{|#1|}}

\newcommand{\Pirr}{\mathsf{Pirr}}

\newcommand{\cl}{{\bf cl}}

\newcommand{\beh}{\mathsf{beh}}

\newcommand{\epito}{\twoheadrightarrow}

\newcommand{\up}{\uparrow}

\newcommand{\monoto}{\rightarrowtail}

\newcommand{\ns}[1]{\mathsf{ns}(#1)}
\newcommand{\nalpha}[1]{\mathsf{natm}(#1)}
\newcommand{\nmu}[1]{\mathsf{nsyn}(#1)}

\newcommand{\src}{{\mathsf{s}}}
\newcommand{\trg}{{\mathsf{t}}}

\newcommand{\R}{\mathcal{R}}

\newcommand{\rev}[1]{{#1}^{\mathsf{r}}}

\newcommand{\Pow}{\mathcal{P}}

\newcommand{\op}{\mathsf{op}}

\newcommand{\SLD}[1]{\mathsf{SLD}(#1)}
\newcommand{\Dep}{\mathbf{Dep}}

\newcommand{\rR}{\mathcal{R}}
\newcommand{\rF}{\mathcal{F}}
\newcommand{\rD}{\mathcal{D}}
\newcommand{\rS}{\mathcal{S}}
\newcommand{\rT}{\mathcal{T}}

\newcommand{\rG}{\mathcal{G}}
\newcommand{\rH}{\mathcal{H}}
\newcommand{\rI}{\mathcal{I}}

\newcommand{\rP}{\mathcal{P}}
\newcommand{\rQ}{\mathcal{Q}}

\newcommand{\rDR}[1]{\mathcal{DR}_{#1}}

\newcommand{\LD}[1]{\mathsf{LD}(#1)}
\newcommand{\BLD}[1]{\mathsf{BLD}(#1)}
\newcommand{\DLD}[1]{\mathsf{DLD}(#1)}
\newcommand{\BLRD}[1]{\mathsf{BLRD}(#1)}

\newcommand{\pOp}{\mathsf{op}}
\newcommand{\Open}{\mathsf{Open}}
\newcommand{\dfa}[1]{\mathsf{dfa}(#1)}

\newcommand{\simple}[1]{\mathsf{simple}(#1)}
\newcommand{\jslReach}[1]{\mathsf{reach}(#1)}

\newcommand{\inte}{\mathrm{\bf in}}

\newcommand{\xto}[1]{\xrightarrow{#1}}
\newcommand{\rsc}[1]{{\mathsf{rsc}}(#1)}
\newcommand{\NFA}{{\bf NFA}}

\newcommand{\ts}[1]{\mathsf{ts}(#1)}
\newcommand{\langs}[1]{\mathsf{langs}(#1)}

\newcommand{\Nleq}{\mathsf{Nleq}}

\newcommand{\dr}{\mathrm{dr}}

\DeclarePairedDelimiter\bra{\langle}{\rvert}
\DeclarePairedDelimiter\ket{\lvert}{\rangle}
\DeclarePairedDelimiterX\braket[2]{\langle}{\rangle}{#1 \delimsize\vert #2}

\newcommand{\Syn}[1]{\mathsf{syn}({#1})}

\newcommand{\ang}[1]{\langle #1 \rangle}

\newcommand{\Land}{\bigwedge}
\newcommand{\Lor}{\bigvee}
\newcommand{\To}{\Rightarrow}
\newcommand{\oT}{\Leftarrow}

\mathchardef\ordinarycolon\mathcode`\:
\mathcode`\:=\string"8000
\begingroup \catcode`\:=\active
  \gdef:{\mathrel{\mathop\ordinarycolon}}
\endgroup

\mathchardef\hyph="2D
\newcommand{\dash}{\mathord{-}}

\addto\extrasbritish{
}

\usepackage[nomargin,marginclue,footnote]{fixme}
\FXRegisterAuthor{sm}{asm}{SM}
\FXRegisterAuthor{hu}{ahu}{HU}
\FXRegisterAuthor{rm}{arm}{RM}

\numberwithin{equation}{section}
\usepackage{cite}
\usepackage{microtype}

\title{Syntactic Minimization of\\Nondeterministic Finite Automata}
\titlerunning{Syntactic Minimization of Nondeterministic Finite Automata}

\author{Robert S.~R. Myers}{London, United Kingdom}{me.robmyers@gmail.com}{}{}

\author{Henning Urbat}{Friedrich-Alexander-Universität Erlangen-Nürnberg, Germany}{henning.urbat@fau.de}{https://orcid.org/0000-0002-3265-7168}%
{Supported by Deutsche Forschungsgemeinschaft (DFG) under project SCHR~1118/15-1}

\authorrunning{R.~S.~R.~Myers and  H.~Urbat}

\Copyright{R.~S.~R.~Myers and H.~Urbat}

\hideLIPIcs  
\begin{document}
\maketitle
\begin{abstract}
  Nondeterministic automata may be viewed as succinct programs implementing deterministic automata, i.e.\ complete specifications. Converting a given deterministic automaton into a small nondeterministic one is known to be computationally very hard; in fact, the ensuing decision problem is $\PSPACE$-complete. This paper stands in stark contrast to the status quo. We restrict attention to \emph{subatomic} nondeterministic automata, whose individual states accept unions of syntactic congruence classes. They are general enough to cover almost all structural results concerning nondeterministic state-minimality. We prove that converting a monoid recognizing a regular language into a small subatomic acceptor corresponds to an $\NP$-complete problem. The $\NP$ certificates are  solutions of simple equations involving relations over the syntactic monoid. We also consider the subclass of \emph{atomic} nondeterministic automata introduced by Brzozowski and Tamm. Given a deterministic automaton and another one for the reversed language, computing small atomic acceptors is shown to be $\NP$-complete with analogous certificates. Our complexity results emerge from an algebraic characterization of (sub)atomic acceptors in terms of deterministic automata with semi\-lattice structure, combined with an equivalence of categories leading to succinct representations.
\end{abstract}

\section{Introduction}
Regular languages arise from a multitude of different perspectives: operationally via finite-state machines, model-theoretically via monadic second-order logic, and algebraically via finite monoids. In practice, \emph{deterministic} finite automata (dfas) and \emph{nondeterministic} finite automata (nfas) are two of the most common representations. Although the former may be exponentially larger than the latter, there is no known efficient procedure for converting dfas into small nfas, e.g.\ state-minimal ones. Jiang and Ravikumar proved the corresponding decision problem (does an equivalent nfa with a given number of states exist?) to be $\PSPACE$-complete~\cite{jk91,MinNfaHard1993}, suggesting that exhaustively enumerating candidates is necessary. One possible strategy towards tractability is to restrict the target automata to suitable \emph{subclasses} of nfas. The challenge is to identify subclasses permitting more efficient computation (e.g.\ lowering the $\PSPACE$ bound to an $\NP$ bound, enabling the use of SAT solvers), while still being general enough to cover succinct acceptors of regular languages.

In our present paper we will show that the class of \emph{subatomic nfas} naturally meets the above requirements. An nfa accepting the language $L$ is {subatomic} if each individual state accepts a union of syntactic congruence classes of $L$. In recent work~\cite{mmu21} we observed that almost all known results on the structure of small nfas, e.g.\ for unary~\cite{MinNfaUnary91,chrobak1986finite}, bideterministic~{\cite{TAMM2004135}}, topological~\cite{ammu14_2} and biRFSA languages~\cite{BiRFSA2009}, implicitly construct small subatomic nfas. This firmly indicates that the latter form a rich class of acceptors despite their seemingly restrictive definition, i.e.\ in many settings computing small nfas amounts to computing small subatomic ones. Restricting to subatomic nfas yields useful additional structure; in fact, their theory is tightly linked to the algebraic theory of regular languages and the representation theory of monoids. This suggests an \emph{algebraic} counterpart of the dfa to nfa conversion problem: given a finite {monoid} recognizing some regular language, compute an equivalent small subatomic nfa. Denoting its decision version (does an equivalent subatomic nfa with a given number of states exist?) by $\MONTOSUBATOMICNFA$,  our main result is:

\begin{theorem*}
The problem $\MONTOSUBATOMICNFA$ is $\NP$-complete.
\end{theorem*}
In addition we also investigate \emph{atomic nfas}, a subclass of subatomic nfas earlier introduced by Brzozowski and Tamm~\cite{TheoryOfAtomataBrzTamm2014}. Similar to the subatomic case, their specific structure naturally invokes the problem of converting a pair of dfas accepting mutually reversed languages into a small atomic nfa. Denoting its decision version by $\DFA+\rev{\DFA}\to \NFA_{\mathbf{atm}}$, we get:
\begin{theorem*}
The problem $\DFA+\rev{\DFA}\to \NFA_{\mathbf{atm}}$ is $\NP$-complete.
\end{theorem*}
The short certificates witnessing that both problems are in $\NP$ are solutions of \emph{equations} involving relations over the syntactic congruence or the Nerode left congruence, respectively. 

The above two theorems sharply contrast the $\PSPACE$-completeness of the general dfa to nfa conversion problem, but also previous results on its sub-$\PSPACE$ variants. The latter are either concerned with particular regular languages such as finite or unary ones~\cite{MinNfaUnary91,gh07}, or with target nfas admitting only very weak forms of nondeterminism, such as unambiguous automata~\cite{MinNfaHard1993} or dfas with multiple initial states~\cite{malcher04}. In contrast, our present work applies to all regular languages and the restriction to (sub)atomic nfas is a purely \emph{semantic} one.

Our results are fundamentally based upon a {category-theoretic perspective} on atomic and subatomic acceptors. At its heart are two equivalences of categories as indicated below:
\[ \JSL^\op_\mathbf{f} \xleftrightarrow[\text{~Structure theory~}]{\simeq} \JSLf \xleftrightarrow[\text{~Complexity theory~}]{\simeq} \Dep.\]
 As shown in~\cite{mmu21}, the \emph{structure theory} of (sub)atomic nfas emerges by interpreting them as dfas endowed with semilattice structure, and relating them to their dual automata under the familiar self-duality of the category $\JSLf$ of finite semilattices. Similarly, the \emph{complexity theory} of (sub)atomic nfas developed in the present paper rests on the equivalence between $\JSLf$ and a category $\Dep$ (see~\autoref{def:dep}) that yields succinct relational representations of finite semilattices by their irreducible elements. To derive the $\NP$-completeness theorems, we reinterpret semilattice automata associated to (sub)atomic nfas inside $\Dep$. We regard this conceptually simple and natural categorical approach as a key contribution of our paper.


\section{Atomic and Subatomic NFAs}\label{sec:prelim}
We start by setting up the notation and terminology used in the rest of the paper, including the key concept of a (sub)atomic nfa that underlies our complexity results. Readers are assumed to be familiar with basic category~\cite{maclane}.

\medskip\noindent\textbf{Semilattices.} A \emph{(join-)semilattice} is a poset $(S,\leq_S)$ in which every finite subset $X\seq S$ has a least upper bound (a.k.a.\ join) ${\bigvee} X$. A \emph{morphism} between semilattices is a map preserving finite joins. If $S$ is finite as we often assume, every subset $X\seq S$ also has a greatest lower bound (a.k.a.\ meet) $\bigwedge X$, given by the join of its lower bounds. In particular, $S$ has a least element $\bot_S=\bigvee \emptyset$ and a greatest element $\top_S=\bigwedge \emptyset$. An element $j\in S$ is \emph{join-irreducible} if $j=\bigvee X$ implies $j\in X$ for every subset $X\seq S$. Dually, $m\in S$ is \emph{meet-irreducible} if $m=\bigwedge X$ implies $m\in X$. We put
 \[J(S) = \{\, j\in S \;:\;\text{$j$ is join-irreducible} \,\} \quad\text{and}\quad M(S) = \{\, m\in S \;:\;\text{$m$ is meet-irreducible} \,\}.\] 
Note $\bot_S\not\in J(S)$ and $\top_S\not\in M(S)$. The join-irreducibles form the least set of \emph{join-generators} of $S$, i.e.\ every element of $S$ is a join of elements from $J(S)$, and every other subset $J\seq S$ with that property contains $J(S)$. Dually, $M(S)$ is the least set of \emph{meet-generators} of $S$.

Let $\mathbb{2}=\{0,1\}$ be the two-element semilattice with $0\leq 1$. Morphisms $i\colon \mathbb{2}\to S$ correspond to elements of $S$ via $i\mapsto i(1)$. Morphisms $f\colon S \to \mathbb{2}$ correspond to \emph{prime filters} via $f\mapsto f^{-1}[1]$. If $S$ is finite, these are precisely the subsets  $F_{s_0}=\{ s\in S:s\not\leq_S s_0 \}$ for any $s_0\in S$.

We denote by $\JSL$ the category of join-semilattices and their morphisms. Its full subcategory $\JSLf$ of finite semilattices is self-dual~\cite{Johnstone1982}: there is an equivalence functor \[\JSL^\op_\mathbf{f}\xra{\simeq}\JSLf\] mapping $(S,\leq_S)$ to the \emph{opposite semilattice} $S^\op = (S,\geq_S)$ obtained by reversing the order, and a morphism $f\colon S\to T$ to the morphism $f_*\colon T^\op\to S^\op$ sending $t\in T$ to the $\leq_S$-greatest element $s\in S$ with $f(s)\leq_T t$. Thus, $f$ and $f_*$ satisfy the adjoint relationship 
\[f(s)\leq_T t \qquad \text{iff}\qquad s\leq_S f_*(t)\]
for all $s\in S$ and $t\in T$. The morphism $f$ is injective (equivalently a $\JSLf$-monomorphism) iff $f_*$ is surjective (equivalently a $\JSLf$-epimorphism).
%

\medskip\noindent \textbf{Relations.} A \emph{relation} between sets $X$ and $Y$ is a subset $\rR\seq X\times Y$. We write $\rR(x,y)$ if $(x,y)\in \rR$. For $x\in X$ and $A\seq X$ we put \[\rR[x] = \{\, y\in Y: \rR(x,y)\,\}\qquad\text{and}\qquad \rR[A] = \bigcup_{x\in A} \rR[x].\] The \emph{converse} of $\rR$ is the relation $\breve{\rR} \seq Y\times X$ (alternatively $\rR\spbreve$) where $\breve{\rR}(y,x)$ iff $\rR(x,y)$ for $x\in X$ and $y\in Y$. The \emph{composite} of $\rR\seq X\times Y$ and $\rS\seq Y\times Z$ is the relation $\rR;\rS\seq X\times Z$ where $\rR(x,z)$ iff there exists $y\in Y$ with $\rR(x,y)$ and $\rS(y,z)$. Let $\Rel$ denote the category whose objects are sets and whose morphisms are relations with the above composition. The identity morphism on $X$ is the identity relation $\id_X\seq X\times X$ with $\id_X(x,y)$ iff $x=y$.

A \emph{biclique} of a relation $\rR\seq X\times Y$ is subset of the form $B_1 \times B_2\seq \rR$, where $B_1\seq X$ and $B_2\seq Y$. A set $\C$ of bicliques forms a \emph{biclique cover} if $\rR=\bigcup \C$. The \emph{bipartite dimension} of $\rR$, denoted $\dim{\rR}$, is the minimum cardinality of any biclique cover. 

\medskip\noindent \textbf{Languages.} Let $\Sigma^*$ be the set of finite words over an alphabet $\Sigma$ including the empty word $\epsilon$. A \emph{language} is a subset $L$ of $\Sigma^*$. We let $\ol{L} = \Sigma^*\setminus L$ denote the \emph{complement} and $\rev{L} = \{\rev{w}: w\in L \}$ the \emph{reverse} of $L$, where $\rev{\epsilon} = \epsilon$ and $\rev{w} = a_n\ldots a_1$ for $w=a_1\ldots a_n$. The \emph{left derivatives} and \emph{two-sided derivatives} of $L$ are, respectively, given by
$u^{-1}L = \{w\in \Sigma^* : uw\in L\}$ and $u^{-1}Lv^{-1} = \{ w\in \Sigma^*: uwv\in L \}$ for $u,v\in \Sigma^*$; moreover for $U\seq \Sigma^*$ put $U^{-1}L = \bigcup_{u\in U} u^{-1}L$. 
For each fixed $L \subseteq \Sigma^*$, the following sets of languages will play a prominent role:
\[
  \LD{L}\seq \SLD{L}\seq \BLD{L}\seq \BLRD{L}
\]
where $\LD{L} = \{  u^{-1} L : u \in \Sigma^* \}$ is the set of left derivatives, and $\SLD{L}$, $\BLD{L}$, $\BLRD{L}$ denote its closure under finite unions, all set-theoretic boolean operations, and all set-theoretic boolean operations and two-sided derivatives, respectively. The final three form $\cup$-semilattices, and the final two are boolean algebras w.r.t.\ the set-theoretic operations.

\smallskip
A language $L$ is \emph{regular} if $\LD{L}$ is a finite set; then the other three sets are finite too. The finite semilattices $\SLD{L}$ and $\SLD{\rev{L}}$ are related by the fundamental isomorphism
\begin{equation}
  \label{eq:drl-iso}
  \dr_L\colon [\SLD{\rev{L}}]^\op \xra{\cong} {\SLD{{L}}},\qquad K\mapsto (\ol{\rev{K}})^{-1}{L},
\end{equation}
see \cite[Proposition 3.13]{mmu21}. Equivalently, the map $\dr_L$ sends $V^{-1}\rev{L}\in \SLD{\rev{L}}$ to the largest element of $\SLD{L}$ disjoint from $\rev{V}$. It is closely connected to the \emph{dependency relation} of $L$,\begin{equation}\label{eq:DRl} \rDR{L}\seq \LD{L}\times \LD{\rev{L}},\qquad \rDR{L}(u^{-1}L,v^{-1}\rev{L})~\mathrel{\vcentcolon\Longleftrightarrow}\ u\rev{v}\in L \quad \text{for $u,v\in \Sigma^*$}.
\end{equation}
In fact, by \cite[Theorem 3.15]{mmu21} we have  \begin{equation}\label{eq_DR_vs_dr} \rDR{L}(u^{-1}L,v^{-1}\rev{L})\iff u^{-1}L\not\seq \dr_L(v^{-1}\rev{L})\qquad \text{for $u,v\in \Sigma^*$.} 
\end{equation}
Since the boolean algebra $\BLD{L}$ is generated by the left derivatives of $L$, its atoms (= join-irreducibles) are the congruence classes of the \emph{Nerode left congruence} $\sim_L\ \seq \Sigma^*\times \Sigma^*$,
\begin{equation}\label{eq:nerodecong} u\sim_L v\quad\text{iff}\quad \forall x\in \Sigma^*: u\in x^{-1}L\Lra v\in x^{-1}L \quad\text{iff}\quad (\rev{u})^{-1}\rev{L}=(\rev{v})^{-1}\rev{L}.\end{equation}
Note that this relation is left-invariant, i.e. $u\sim_L v$ implies $wu\sim_L wv$ for all $w\in \Sigma^*$.
 
Similarly, the atoms of $\BLRD{L}$ are the congruence classes of the \emph{syntactic congruence} $\equiv_L\ \seq \Sigma^*\times \Sigma^*$, i.e. the monoid congruence on the free monoid $\Sigma^*$ defined by
\begin{equation}\label{eq:syncong}u\equiv_L v\quad\text{iff}\quad \forall x,y\in \Sigma^*: u\in x^{-1}Ly^{-1}\Lra v\in x^{-1}Ly^{-1}.\end{equation}
The quotient monoid $\Syn{L}=\Sigma^*/{\equiv_L}$ is called the \emph{syntactic monoid} of $L$, and the canonical map $\mu_L\colon \Sigma^*\epito \Syn{L}$ sending $u\in \Sigma^*$ to its congruence class $[u]_{\equiv_L}$ is the \emph{syntactic morphism}.


\medskip\noindent \textbf{Automata.} Fix a finite alphabet $\Sigma$. A \emph{nondeterministic finite automaton} (a.k.a.\ \emph{nfa}) $N=(Q,\delta,I,F)$ consists of a finite set $Q$ (the states), relations $\delta=(\delta_a\seq Q\times Q)_{a\in \Sigma}$ (the transitions), and sets $I,F\seq Q$ (the initial states and final states). We write $q_1\xra{a}q_2$ whenever $q_2\in \delta_a[q_1]$. The language $L(N,q)$ \emph{accepted} by a state $q\in Q$ consists of all words $w\in \Sigma^*$ such that $\delta_w[q]\cap F\neq \emptyset$, where $\delta_w\seq Q\times Q$ is the extended transition relation $\delta_{a_1};\ldots;\delta_{a_n}$ for $w=a_1\ldots a_n$ and $\delta_\epsilon = \id_Q$. The language \emph{accepted} by $N$ is defined $L(N)=\bigcup_{i\in I} L(N,i)$.

An nfa $N$ is a \emph{deterministic finite automaton} (a.k.a.\ \emph{dfa}) if $I=\{q_0 \}$ is a singleton set and each transition relation is a function $\delta_a\colon Q\to Q$. A dfa is a \emph{$\JSL$-dfa} if $Q$ is a finite semilattice, each $\delta_a\colon Q\to Q$ is a semilattice morphism, and $F\seq Q$ forms a prime filter. It is often useful to represent a $\JSL$-dfa in terms of morphisms
\[
  \mathbb{2}\xra{i}Q\xra{\delta_a}Q\xra{f} \mathbb{2}
\]
where $i$ is the unique morphism with $i(1)=q_0$ and $f$ is given by $f(q)=1$ iff $q\in F$.
A $\JSL$-dfa \emph{morphism} from $A=(Q,\delta,i,f)$ to $A'=(Q',\delta',i',f')$ is a $\JSLf$-morphism $h\colon Q\to Q'$ preserving transitions via $h\circ \delta_a = \delta_a'\circ h$, preserving the initial state via $i'=h\circ i$, and both preserving and reflecting the final states via $f=f'\circ h$. Equivalently, $h$ is a dfa morphism that is also a semilattice morphism, so in particular $L(A)=L(A')$. If $Q$ is a subsemilattice of $Q'$ and $h\colon Q\monoto Q'$ is the inclusion map, then $A$ is called a \emph{sub $\JSL$-dfa} of $A'$.

Fix a regular language $L$. Viewed as a $\cup$-semilattice, $\BLRD{L}$ carries the structure of a $\JSL$-dfa with transitions $K\xra{a}a^{-1}K$, initial state $L$, and finals $\{ K : \epsilon \in K \}$. This restricts to sub $\JSL$-dfas $\BLD{L}$ and $\SLD{L}$. Moreover $\LD{L}$ forms a sub-dfa of $\SLD{L}$, well-known~\cite{BrzozowskiDRE1964} to be the \emph{state-minimal dfa} for $L$, so we denote it by $\dfa{L}$. The syntactic monoid $\Syn{L}$ is isomorphic to the \emph{transition monoid} of $\dfa{L}$, i.e.\ the monoid of all extended transition maps $\delta_w\colon \LD{L}\to \LD{L}$ ($w\in \Sigma^*$) with multiplication given by composition~\cite{pin20}. 

Analogously $\SLD{L}$ is the \emph{state-minimal $\JSL$-dfa} for $L$. Up to isomorphism, it is the unique $\JSL$-dfa for $L$ that is \emph{$\JSL$-reachable} (i.e.\ every state is a join of states reachable from the initial state via transitions) and \emph{simple} (i.e.\ distinct states accept distinct languages).

Nfas, dfas and $\JSL$-dfas are expressively equivalent and accept precisely the regular languages. In particular, to every $\JSL$-dfa $A=(Q,\delta,q_0,F)$ one can associate an equivalent nfa $J(A)$, the \emph{nfa of join-irreducibles}~\cite{arbib_manes_1975,ammu14_2,mamu15}. Its states are given by the set $J(Q)$ of join-irreducibles of $Q$; for any $q_1,q_2\in J(Q)$ and $a\in \Sigma$ there is a transition $q_1\xra{a} q_2$ in $J(A)$ iff $q_2\leq_Q \delta_a(q_1)$; a state $q\in J(Q)$ is initial iff $q\leq_S q_0$, and final iff $q\in F$. For any $q\in J(Q)$, we have $L(A,q)=L(J(A),q)$. The \emph{canonical residual finite state automaton}~\cite{ResidFSA2001} for a regular language $L$ is given by $N_L=J(\SLD{L})$, the nfa of join-irreducibles of its minimal $\JSL$-dfa.

\medskip\noindent \textbf{Atomic and subatomic nfas.}
An nfa accepting the language $L\seq\Sigma^*$ is called \emph{atomic}~\cite{TheoryOfAtomataBrzTamm2014} if each state accepts a language from $\BLD{L}$, and \emph{subatomic}~\cite{mmu21} if each state accepts a language from $\BLRD{L}$. The \emph{nondeterministic atomic complexity} $\nalpha{L}$ of a regular language $L$ is the least number of states of any atomic nfa accepting $L$. The \emph{nondeterministic syntactic complexity} $\nmu{L}$ is the least number of states of any subatomic nfa accepting $L$. Subatomic nfas are intimately connected to syntactic monoids: the atoms of $\BLRD{L}$ are the elements of $\Syn{L}$, so an nfa accepting $L$ is subatomic iff its individual states accept unions of syntactic congruence classes. Additionally $\nmu{L}$ can be characterized via \emph{boolean representations} of $\Syn{L}$, i.e.\ monoid morphisms $\rho\colon \Syn{L}\to \JSLf(S,S)$ into the endomorphisms of a finite semilattice~\cite{mmu21}. For a detailed exposition we refer to \emph{op.\ cit.}

These complexity measures are related to the \emph{nondeterministic state complexity} $\ns{L}$, i.e.\ the least number of states of any (unrestricted) nfa accepting $L$. In particular,
\begin{equation}\label{eq:complexity_ineq}\dim{\rDR{L}}\leq \ns{L}\leq \nmu{L}\leq \nalpha{L}.
\end{equation}
The first inequality is due to Gruber and Holzer~\cite{LowerBoundsHard} (see also \cite[Theorem 4.8]{mmu21} for a purely algebraic proof), while the others arise by restricting admissible nondeterministic acceptors.

Importantly, small atomic and subatomic nfas can be characterized in terms of $\JSL$-dfas. The following theorem involves two commuting diagrams of semilattice morphisms, whose lower and upper paths are the canonical $\JSL$-dfas described earlier.

\begin{theorem}\label{thm:na_nmu_char}
Let $L\seq \Sigma^*$ be a regular language.
\begin{enumerate}
\item $\nalpha{L}$ is the least number $k$ such that there exists a finite semilattice $S$ with $|J(S)|\leq k$ and $\JSLf$-morphisms $p,q$ and $\tau_a$ ($a\in \Sigma$) making the left-hand diagram below commute.
\item $\nmu{L}$ is the least number $k$ such that there exists a finite semilattice $S$ with $|J(S)|\leq k$ and $\JSLf$-morphisms $p,q$ and $\tau_a$ ($a\in \Sigma$) making the right-hand diagram below commute.
\end{enumerate}
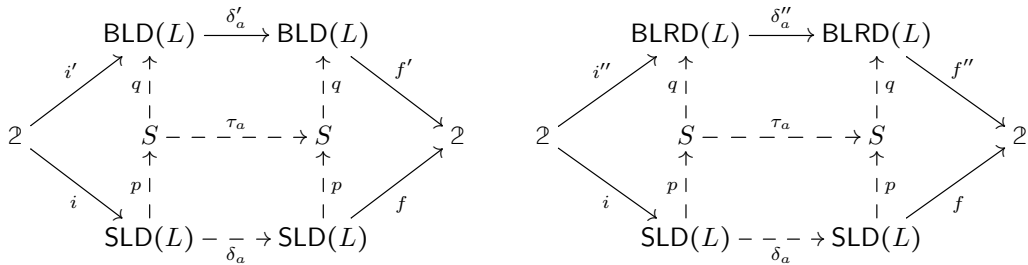
\[
\xymatrix{
& \BLD{L} \ar[r]^{\delta_a'} & \BLD{L} \ar[dr]^{f'} & \\
\mathbb{2} \ar[ur]^{i'} \ar[dr]_{i} & S \ar@{-->}[u]^q \ar@{-->}[r]^{\tau_a} & S \ar@{-->}[u]_q & \mathbb{2} \\
& \SLD{L} \ar@{-->}[u]^p \ar@{-->}[r]_{\delta_a} & \SLD{L} \ar@{-->}[u]_p \ar[ur]_{f} & 
}
\qquad
\xymatrix{
& \BLRD{L} \ar[r]^{\delta_a''} & \BLRD{L} \ar[dr]^{f''} & \\
\mathbb{2} \ar[ur]^{i''} \ar[dr]_{i} & S \ar@{-->}[u]^q \ar@{-->}[r]^{\tau_a} & S \ar@{-->}[u]_q & \mathbb{2} \\
& \SLD{L} \ar@{-->}[u]^p \ar@{-->}[r]_{\delta_a} & \SLD{L} \ar@{-->}[u]_p \ar[ur]_{f} & 
}
\]
\end{theorem}

\begin{proof}
We only prove part (1), the proof of (2) being completely analogous. 

\medskip\noindent Suppose there exists a finite semilattice $S$ with $|J(S)|=k$ and $\JSLf$-morphisms $p,q$ and $(\tau_a)_{a \in \Sigma}$ making the left diagram commute. Then $A=(S,\tau,p\circ i, f'\circ q)$ is a $\JSL$-dfa and $p\colon \SLD{L}\to A$ and $q\colon A\to \BLD{L}$ are $\JSL$-dfa morphisms. Since $\JSL$-dfa morphisms preserve the accepted language, and every state $K\in \BLD{L}$ accepts the language $K$, it follows that $A$ accepts $L$ and every state of $A$ accepts a language from $\BLD{L}$. Thus the nfa $J(A)$ of join-irreducibles corresponding to $A$ is an atomic nfa for $L$ with $k$ states.

\medskip\noindent Conversely, assume $N=(Q,\delta,I,F)$ is a $k$-state atomic nfa  accepting $L$. Form the $\cup$-semilattice $S=\langs{N}$ of all languages $L(N,X)$ accepted by subsets $X\seq Q$. Note that $\SLD{L}\seq S\seq \BLD{L}$: the first inclusion holds because $u^{-1}L=L(N,\delta_w[I])\in S$ for every $u\in \Sigma^*$, and the second one because $N$ is atomic. We define the semilattice endomorphisms 
\[ \tau_a\colon S\to S\qquad\text{by}\qquad \tau_a(K)=a^{-1}K\quad\text{for $K\in S$}, \]
Letting $p\colon \SLD{L}\monoto S$ and $q\colon S\monoto \BLD{L}$ denote the inclusions, the left diagram commutes. Moreover $|J(S)|\leq k$ since $S$ is join-generated by the elements $L(N,q)$ for $q\in Q$.
\end{proof}

\section{Representing Finite Semilattices as Finite Relations}\label{sec:dep}


We have seen that atomic and subatomic nfas amount to certain dfas with semi\-lattice structure. To obtain our $\NP$-completeness results concerning the computation of small (sub)atomic acceptors we will study succinct representations of the corresponding $\JSLf$-diagrams from \autoref{thm:na_nmu_char}. For this purpose, we start with the following key observation:
\begin{quote}
  Any finite semilattice $S$ is completely determined by its \emph{poset of irreducibles}~\cite{MarkowskyLat1975}, i.e.\ the relation $\not\leq_S\ \seq J(S)\times M(S)$ between join-irreducibles and meet-irreducibles.
\end{quote}
 We now prove that this extends to an \emph{equivalence} between
 the category $\JSLf$ of finite semilattices and another category called $\Dep$. Its objects are the relations between finite sets and its morphisms represent semilattice morphisms as relations. The equivalence is inspired by Moshier's \emph{categories of contexts}~\cite{jipsen12,moshier16} and will serve as the conceptual basis of our work.

\begin{defn}[The category of dependency relations]
  \label{def:dep}
  The objects of the category $\Dep$ are the relations $\rR \seq \rR_\src \times \rR_\trg$ between finite sets. Far less obviously,
  \begin{quote}
    a morphism $\rP\colon \rR \to \rS$ is a relation $\rP \seq \rR_\src \times \rS_\trg$ that factorizes through $\rR$ and $\rS$, i.e.\ the left $\Rel$-diagram below commutes for some $\rP_l \seq \rR_\src \times \rS_\src$ and $\rP_u  \seq \rS_\trg \times \rR_\trg$.
  \end{quote}
  The identity morphism for $\rR$ is $\id_\rR=\rR$, see the central diagram below. 
  The composite $\rP\fatsemi \rQ\colon \rR\to \rT$ of $\rP\colon \rR\to \rS$ and $\rQ\colon \rS\to \rT$ is any of the five equivalent relational compositions starting from the bottom left corner and ending at the top right corner of the rightmost diagram below; that is,
$\rP \fatsemi \rQ := \rP_l ; \rQ_l ; \rT= \rP_l ; \rQ= \rP_l ; \rS ; \rQ_u\spbreve = \rP ; \rQ_u\spbreve = \rR ; \rP_u\spbreve ;\rQ_u\spbreve$. (Note that we use the symbol $\fatsemi$ for composition in $\Dep$ and $;$ for composition in $\Rel$, and recall that $(\dash)\spbreve$ denotes the converse relation.)
  \[
  \xymatrix{
    \rR_\trg \ar@{-->}[rr]^{\rP_u\spbreve}  && \rS_\trg
    \\
    \rR_\src \ar[u]^-{\rR} \ar[urr]^-\rP \ar@{-->}[rr]_{\rP_l} && \rS_\src \ar[u]_-{\rS}
  }
\quad
  \xymatrix{
    \rR_\trg \ar@{-->}[rr]^{\id\spbreve}  && \rR_\trg
    \\
    \rR_\src \ar[u]^-{\rR} \ar[urr]^-\rR \ar@{-->}[rr]_{\id} && \rR_\src \ar[u]_-{\rR}
  } 
\quad 
  \xymatrix{
    \rR_\trg \ar@{-->}[rr]^{\rP_u\spbreve}  && \rS_\trg \ar@{-->}[rr]^{\rQ_u\spbreve} && \rT_\trg
    \\
    \rR_\src \ar[u]^-{\rR} \ar[urr]^-\rP \ar@{-->}[rr]_{\rP_l} && \rS_\src \ar[u]_-{\rS} \ar[urr]^-\rQ \ar@{-->}[rr]_{\rQ_l} && \rT_\src \ar[u]_-{\rT}
  }
  \]
\end{defn}
One readily verifies that $\Dep$ is a well-defined category; in particular, the composition is independent of the choice of the lower and upper witnesses $(\dash)_l$ and $(\dash)_u$. 

\begin{rem}\label{rem:dep-witnesses}
  \begin{enumerate}
    \item Using the converse upper witness may seem strange. Although technically unnecessary, it fits the self-duality of $\Dep$ taking the converse on objects and morphisms. Moreover $f ; \nleq_T\ =\ \nleq_S ; f_*\spbreve$ for any $\JSLf$-morphism $f\colon S \to T$ via the adjoint relationship; that is, $f$ induces a $\Dep$-morphism from $\not\leq_S$ to $\not\leq_T$ with lower witness $f$ and upper witness $f_*$.
    \item The witnesses of a $\Dep$-morphism $\rP \colon \rR \to \rS$ are closed under unions. The maximal lower witness $\rP_-\seq \rR_\src\times \rS_\src$ is given by
      \[\rP_-(x, y) ~\mathrel{\vcentcolon\Longleftrightarrow}\ \rS[y] \subseteq \rP[x]\qquad\text{for}\qquad  x\in \rR_\src,\, y\in \rS_\src,\]
       and the maximal upper witness $\rP_+\seq \rS_\trg\times \rR_\trg$ by
      \[\rP_+(y, x) ~\mathrel{\vcentcolon\Longleftrightarrow}\ \breve{\rR}[x] \subseteq \breve{\rP}[y] \qquad\text{for}\qquad x\in \rR_\trg,\, y\in \rS_\trg.\]
  \end{enumerate}
\end{rem}

\begin{theorem}[Fundamental equivalence]\label{thm:jsl_vs_dep}
The categories $\JSLf$ and $\Dep$ are equivalent.

\smallskip
\begin{enumerate} 
\item The equivalence functor $\Pirr\colon \JSLf \to \Dep$ maps a finite semilattice $S$ to the $\Dep$-object
\[
  \Pirr(S)\;:=\;\not\leq_S \;\seq \;J(S)\times M(S),
\]
 and a $\JSLf$-morphism $f\colon S\to T$ to the $\Dep$-morphism
\[
  \Pirr(f)\colon  \Pirr(S)\to \Pirr(T),\qquad \Pirr(f)(j,m) :\Lra f(j)\not\leq_T m  \quad \text{for $j\in J(S)$, $m\in M(T)$}.
\]
\item The inverse $\Open\colon \Dep \to \JSLf$ maps a $\Dep$-object $\rR$ to its \emph{semilattice of open sets} 
\[
  \Open(\rR) \;:=\; (\{ \rR[X]: X\seq \rR_\src \},\seq),
\]
and a $\Dep$-morphism $\rP\colon \rR\to \rS$ to the $\JSLf$-morphism
\[ \Open(\rP)\colon \Open(\rR)\to \Open(\rS),\qquad \Open(\rP)(O) := \rP_+\spbreve[O] \quad\text{for $O\in \Open(\rR)$},
 \]
where $\rP_+\seq \rS_\trg\times\rR_\trg$ is the maximal upper witness of $\rP$.
\end{enumerate}
\end{theorem}

\begin{rem}\label{rem:join_meet_generators}
In the definition of $\Pirr(S)$ one may replace $J(S)$ and $M(S)$ by any two sets $J,M\seq S$ of join- and meet-generators modulo $\Dep$-isomorphism. Indeed, since the equivalence functor $\Open$ reflects isomorphisms, this follows immediately from the $\JSLf$-isomorphism $\Open(\not\leq_S\cap\ {J\times M})\cong \Open(\not\leq_S\cap\ {J(S)\times M(S)})$ given by $O\mapsto O\cap M(S)$.
\end{rem}

\begin{rem}\label{rem:dep_bip_isos}
  Bijectively relabeling the domain and codomain of a relation defines a $\Dep$-isomorphism, the witnesses being the relabelings.
\end{rem}

\noindent We now show that for every regular language $L$, the semilattices $\SLD{L}$, $\BLD{L}$ and $\BLRD{L}$ equipped with their canonical $\JSL$-dfa structure (see \autoref{sec:prelim}) translate under the equivalence functor $\Pirr$ into familiar concepts from automata theory. The translations are summarized in \autoref{tab1} and explained in Examples \ref{ex:sld-vs-dlr}--\ref{ex:blrd-vs-syn} below. 

\begin{table}[H]
  \small
\centering
\begin{tikzpicture}
\node (table) [inner sep=1pt] {
\begin{tabular}{C{6cm}C{6cm}}
$\JSLf$ & $\Dep$  \\
\hline
$\mathbb{2} \xra{i}  \SLD{L}\xra{\delta_a} \SLD{L} \xra{f} \mathbb{2}$ &  $\id_1 \xra{\rI} \rDR{L} \xra{\rDR{L,a}} \rDR{L} \xra{\rF} \id_1$ \\
  \hdashline
$\mathbb{2} \xra{i'}  \BLD{L}\xra{\delta_a'} \BLD{L} \xra{f'} \mathbb{2}$ & $\id_1 \xra{\rI'} \id_{\Sigma^*/{\sim_L}} \xra{\rD_a'} \id_{\Sigma^*/{\sim_L}} \xra{\rF'} \id_1$ \\
\hdashline
$\mathbb{2} \xra{i''}  \BLRD{L}\xra{\delta_a''} \BLRD{L} \xra{f''} \mathbb{2}$ & $\id_1 \xra{\rI''} \id_{\Syn{L}} \xra{\rD_a''} \id_{\Syn{L}} \xra{\rF''} \id_1$\\
\end{tabular}
};
\draw [rounded corners=.5em] (table.north west) rectangle (table.south east);
\end{tikzpicture}
\caption{Canonical $\JSL$-dfas and their corresponding $\Dep$-structures}
\label{tab1}
\end{table}

\begin{example}[State-minimal $\JSL$-dfa vs.\ dependency relation $\rDR{L}$]
\label{ex:sld-vs-dlr}
Let us start with the observation that $\SLD{L}$ is join-generated by $\LD{L}$ and meet-generated by $\dr_L[\LD{\rev{L}}]$. The latter follows via the fundamental isomorphism \eqref{eq:drl-iso}. Then
\[
  \Pirr(\SLD{L})(u^{-1} L, \dr_L(v^{-1} \rev{L}))\;\;\stackrel{\text{def.}}{\Longleftrightarrow}\;\; u^{-1} L \nsubseteq \dr_L(v^{-1} \rev{L}) \;\;\stackrel{\text{\eqref{eq_DR_vs_dr}}}{\Longleftrightarrow}\;\; \rDR{L}(u^{-1} L, v^{-1}\rev{L})
\]
for every $u^{-1}L\in J(\SLD{L})$ and $v^{-1}\rev{L} \in J(\SLD{\rev{L}})$. 
Thus,
\begin{quote}
  $\Pirr(\SLD{L})$ is a bijective relabeling of $\rDR{L}$ restricted to $J(\SLD{L}) \times J(\SLD{\rev{L}})$.
\end{quote}
By \autoref{rem:join_meet_generators} we know $\Pirr(\SLD{L})$ is isomorphic to the domain-codomain extension $\nsubseteq\ \subseteq \LD{L} \times \dr_L[\LD{L^r}]$ and thus also to the dependency relation $\rDR{L}$ by Remark \ref{rem:dep_bip_isos}. Then the $\JSL$-dfa structure of the semilattice $\SLD{L}$ translates into the category of dependency relations as shown in \autoref{tab1}, where $\id_1$ is the identity relation on $1 = \{ * \}$ and
\[ 
\begin{array}{lllll}
\rI\seq 1\times \LD{\rev{L}}, && \rDR{L,a}\seq \LD{L}\times \LD{\rev{L}}, && \rF\seq \LD{{L}}\times 1, \\
\rI(\ast,v^{-1}\rev{L}) \Lra v \in \rev{L}, && \rDR{L,a}(u^{-1}L,v^{-1}\rev{L}) \Lra  ua\rev{v}\in L, && \rF(u^{-1}{L},\ast)\Lra {u}\in L.
\end{array}
\]
\end{example}

\begin{example}[$\BLD{L}$ vs.\ the Nerode left congruence $\sim_L$]\label{ex:bld-vs-nerode}
  In \autoref{sec:prelim} we observed that the atoms of the boolean algebra $\BLD{L}$ are the congruence classes of the Nerode left congruence. Then the co-atoms are their relative complements, and
  \[
    \Pirr(\BLD{L})([u]_{\sim_L},\ol{[v]_{\sim_L}})
    \xLeftrightarrow{\text{def.}} [u]_{\sim_L}\not\seq \ol{[v]_{\sim_L}}
    \iff [u]_{\sim_L} = [v]_{\sim_L}.
  \]
  By Remark \ref{rem:dep_bip_isos}, we see that $\BLD{L}$ corresponds to the $\Dep$-object $\id_{\Sigma^*/{\sim_L}}$, and its $\JSL$-dfa structure translates into the category of dependency relations as indicated in \autoref{tab1}, where
\[ 
\begin{array}{lll}
  \rI'\seq 1\times \Sigma^*/{\sim_L}, & \rD_a' \seq \Sigma^*/{\sim_L}\times \Sigma^*/{\sim_L}, & \rF'\seq \Sigma^*/{\sim_L}\times 1, \\
  \rI'(\ast,[u]_{\sim_L}) \Lra u\in L, & \rD_a'([u]_{\sim_L},[v]_{\sim_L}) \Lra  [v]_{\sim_L}\seq a^{-1}[u]_{\sim_L}, & \rF'([u]_{\sim_L},\ast)\Lra u\sim_{L} \epsilon.
\end{array}
\]
We note that the above relations induce an nfa
  \[
    (\Sigma^*/{\sim_L}, (\rD_a')_{a \in \Sigma}, \rI'[*], \breve{\rF}'[*])  
    \qquad
    \text{known as the \emph{\'{a}tomaton} for the language $L$~\cite{TheoryOfAtomataBrzTamm2014}.}
  \]
\end{example}

\begin{example}[$\BLRD{L}$ vs. the syntactic monoid $\Syn{L}$]\label{ex:blrd-vs-syn}
Analogously, the boolean algebra $\BLRD{L}$ corresponds to the $\Dep$-object $\id_{\Syn{L}}$. Its semilattice dfa structure translates into the category of dependency relations as shown in \autoref{tab1}, where
\[ 
\arraycolsep=4.0pt
\begin{array}{lll}
\rI''\seq 1\times \Syn{L}, & \rD_a'' \seq \Syn{L}\times \Syn{L}, & \rF''\seq \Syn{L}\times 1, \\
\rI''(\ast,[u]_{\equiv_L}) \Lra u\in L, & \rD_a''([u]_{\equiv_L},[v]_{\equiv_L}) \Lra  [v]_{\equiv_L}\seq a^{-1}[u]_{\equiv_L}, & \rF''([u]_{\equiv_L},\ast)\Lra u\equiv_L \epsilon.
\end{array}
\]
\end{example}

We conclude this section with two lemmas establishing important properties of the equivalence. The first concerns the bipartite dimension of relations (see \autoref{sec:prelim}):
\begin{lemma}\label{lem:dep-iso-preserves-dim}
Let $\rR$ be a relation between finite sets.
\begin{enumerate}
\item $\dim{\rR}$ is the least $|J(S)|$ of any injective $\JSLf$-morphism $m\colon \Open(\rR)\monoto S$.
\item $\dim{\rR}$ is invariant under isomorphism, i.e.\ $\rR\cong \rS$ in $\Dep$ implies $\dim{\rR}=\dim{\rS}$.
\end{enumerate}
\end{lemma}
The second explicitly describes the join- and meet-irreducibles of the semilattice $\Open(\rR)$.

\begin{notation} For $\rR\seq \rR_\src\times \rR_\trg$ we define the following operator on the power set of $\rR_\trg$:
\[ \intop_{\rR}\colon \Pow(\rR_\trg)\to \Pow(\rR_\trg),\qquad Y\;\mapsto\; \bigcup \{ \rR[X] \;:\; X\seq \rR_\src \text{ and } \rR[X]\seq Y \}.  \]
Thus, $\intop_{\rR}(Y)$ is the largest open set of $\rR$ contained in $Y\seq \rR_\trg$.
\end{notation}

\begin{lemma}\label{lem:jirr-mirr-open}
Let $\rR\seq \rR_\src \times \rR_\trg$ be a relation between finite sets.
\begin{enumerate}
\item $J(\Open(\rR))$ consists of all sets $\rR[x]$ ($x\in \rR_\src$) that cannot be expressed as a union of smaller such sets, i.e. $\rR[x]=\bigcup_{i\in I} \rR[x_i]$ implies $\rR[x]=\rR[x_i]$ for some $i\in I$.
\item $M(\Open(\rR))$ consists of all sets $\intop_{\rR}(\rR_\trg\setminus \{y\})$ such that  $\breve{\rR}[y]$ lies in $J(\Open(\breve{\rR}))$.
\end{enumerate}
\end{lemma}

\section{Nuclear Languages and Lattice Languages}\label{sec:nuclear-lattice}

As a further technical tool, we now introduce two classes of regular languages. They are well-behaved w.r.t.\ their small nfas and will emerge at the heart of our $\NP$-completeness proofs in \autoref{sec:complexity}. Their definition rests on the notion of a \emph{nuclear morphism} in $\JSLf$, originating from the theory of symmetric monoidal closed categories \cite{RoweNuclear1988, HiggsRoweNuclear1989}. Recall that a finite semilattice is a \emph{distributive lattice} if $x\wedge (y\vee z) = (x\wedge y) \vee (x\wedge z)$ for all elements $x,y,z$.

\begin{defn}[Nuclear language]
A $\JSLf$-morphism $f\colon S\to T$ is \emph{nuclear} if it factorizes through a finite  distributive lattice, i.e.\ $f=(S\xra{g} D \xra{h}T)$ for some finite distributive lattice $D$ and $\JSLf$-morphisms $g,h$. A regular language $L\seq \Sigma^*$ is \emph{nuclear} if the transition morphisms $\delta_a=a^{-1}(\dash)\colon \SLD{L}\to \SLD{L}$ ($a\in \Sigma$) of its minimal $\JSL$-dfa are nuclear.
\end{defn}

\begin{example}[BiRFSA languages]
A regular language $L$ is \emph{biRFSA}~\cite{BiRFSA2009} if $\rev{(N_L)} \cong {N_{\rev{L}}}$, that is, the canonical residual finite state automata for $L$ and $\rev{L}$ (see \autoref{sec:prelim}) are reverse-isomorphic. In~\cite[Example 5.7]{mmu21} we proved that the biRFSA languages are precisely those whose semilattice $\SLD{L}$ is distributive. Thus biRFSA languages are nuclear.
\end{example}
There is a natural subclass of nuclear languages which need not be biRFSA:

\begin{defn}[Lattice language]\label{def:lattice_langs}
For any $S\in\JSLf$ we define the language $L(S)\seq \Sigma^*$,
\[ \Sigma := \{ \bra{j} : j \in J(S) \} \cup \{ \ket{m} : m \in M(S) \} \qquad\text{and}\qquad L(S) := \bigcap_{j \leq_S m} \overline{\Sigma^* \bra{j}\,\ket{m} \Sigma^*}. \]
Then $\Sigma$ is the disjoint union of $J(S)$ and $M(S)$ (with the notation $\bra{j}$ and $\ket{m}$ used to distinguish between elements of the two summands), and $L(S)$ consists of all words over $\Sigma$ not containing any factor $\bra{j}\,\ket{m}$ with $j\leq_S m$.
\end{defn}


\begin{lemma}\label{lem:lattice_lang_nuclear}
For any $S\in \JSLf$, the language $L(S)$ is nuclear and $S \cong \SLD{L(S)}$.
\end{lemma}

\noindent Crucially, for nuclear and lattice languages some of the relations \eqref{eq:complexity_ineq} hold with equality:
\begin{proposition}\label{prop:lattice-complexity}\label{prop:nuclear-complexity}
\begin{enumerate}
\item If $L$ is a nuclear language then $\ns{L}=\dim{\rDR{L}}$.
\item If $L=L(S)$ is a lattice language then $\nalpha{L} = \nmu{L}=\ns{L}=\dim{\rDR{L}}$.
\end{enumerate}
\end{proposition}
These equalities are the key fact making our reductions in the next section work.

\section{Complexity of Computing Small (Sub)Atomic Acceptors}\label{sec:complexity}
We are ready to present our main complexity results on small (sub)atomic nfas. First we consider the slightly simpler atomic case, phrased as the following decision problem:

\medskip\noindent$\DFA + \rev{\DFA} \to \NFA_{\mathbf{atm}}$\\
\textbf{Input:} Two dfas $A$ and $B$ such that $L(A)=\rev{{L(B)}}$ and a natural number $k$.\\
\textbf{Task:} Decide whether there exists a $k$-state atomic nfa equivalent to $A$, i.e.\ $\nalpha{L(A)}\leq k$.

\begin{rem}\label{rem:dfapair}
Taking mutually reverse dfas $(A,B)$ as input permits an efficient computation of the dependency relation $\rDR{L}\seq \LD{L}\times \LD{\rev{L}}$ of $L=L(A)$. One may assume $A$ and $B$ are minimal dfas, so that their state sets $Q_A$ and $Q_B$ are in bijective correspondence with $\LD{L}$ and $\LD{\rev{L}}$. For $p\in Q_A$ choose some $w_A(p) \in \Sigma^*$ sending the initial state to $p$; analogously choose $w_B(q)\in \Sigma^*$ for $q\in Q_B$. Then $\rDR{L}$ is a bijective relabeling of 
\[ \widetilde{\rDR{L}} \seq Q_A\times Q_B\qquad\text{where} \qquad \widetilde{\rDR{L}}(p,q) ~ \mathrel{\vcentcolon\Longleftrightarrow}\  \text{$A$ \,accepts\, $w_A(p) \rev{w_B(q)}$},  \]
so it is computable in polynomial time from $A$ and $B$. A completely analogous argument applies to the relations $\rI$, $\rDR{L,a}$ and $\rF$ from \autoref{ex:sld-vs-dlr}.
\end{rem}

\begin{theorem}
  \label{thm:atomic_npc}
 The problem $\DFA + \rev{\DFA} \to \NFA_{\mathbf{atm}}$ is $\NP$-complete.
\end{theorem}
We establish the upper and lower bound separately in the next two propositions. Both their proofs are  based on the fundamental equivalence between $\JSLf$ and $\Dep$.
 
\begin{proposition}\label{prop:atomic_np}
The  problem $\DFA + \rev{\DFA} \to \NFA_{\mathbf{atm}}$ is in $\NP$.
\end{proposition}

\begin{proof}
\begin{enumerate}
\item One can check in polynomial time whether a given pair $(A,B)$ of dfas forms a valid input, i.e.\ satisfies $L(A)=\rev{L(B)}$. In fact, this condition is equivalent to $\ol{L(A)}\cap \rev{L(B)}=\ol{L(B)}\cap \rev{L(A)}=\emptyset$. Using the standard methods for complementing dfas and reversing and intersecting nfas, one can construct nfas for  $\ol{L(A)}\cap \rev{L(B)}$ and $\ol{L(B)}\cap \rev{L(A)}$ of size polynomial in $|A|$ and $|B|$, the number of states of $A$ and $B$, and check for emptyness by verifying that no final state is reachable from the initial states.
\smallskip
\item Let $A$ and $B$ be dfas accepting the languages $L$ and $\rev{L}$, respectively, and let $k$ be  a natural number. We claim the following three statements to be equivalent:
\begin{enumerate}
\item There exists an atomic nfa accepting $L$ with at most $k$ states.
\item There exists a finite semilattice $S$ with $|J(S)|\leq k$ and $\JSLf$-morphisms $p,q$ and $\tau_a$ ($a\in \Sigma$) making the left diagram below commute.
\item There exists a $\Dep$-object $\rS\seq \rS_\src\times \rS_\trg$ with $|\rS_\src|\leq k$ and $|\rS_\trg|\leq |B|$ and $\Dep$-morphisms $\rP$, $\rQ$ and $\rT_a$ ($a\in \Sigma$) making the right diagram below commute (cf. \autoref{ex:sld-vs-dlr}/\ref{ex:bld-vs-nerode}).
\end{enumerate}
\begin{equation}\label{eq:equiv-diagrams}
\vcenter{
\xymatrix@C-0.5em{
& \BLD{L} \ar[r]^{\delta_a'} & \BLD{L} \ar[dr]^{f'} & \\
\mathbb{2} \ar[ur]^{i'} \ar[dr]_{i} & S \ar@{-->}[u]^q \ar@{-->}[r]^{\tau_a} & S \ar@{-->}[u]_q & \mathbb{2} \\
& \SLD{L} \ar@{-->}[u]^p \ar@{-->}[r]_{\delta_a} & \SLD{L} \ar@{-->}[u]_p \ar[ur]_{f} & 
}
}
\qquad
\vcenter{
\xymatrix@C-0.5em{
& \id_{\Sigma^*/\sim_L} \ar[r]^{\rD_a'} & \id_{\Sigma^*/{\sim_{L}}} \ar[dr]^{\rF'} & \\
\id_1 \ar[ur]^{\rI'} \ar[dr]_{\rI} & \rS \ar@{-->}[u]^{\rQ} \ar@{-->}[r]^{\rT_a} & \rS \ar@{-->}[u]_{\rQ} & \id_1 \\
& \rDR{L} \ar@{-->}[u]^\rP \ar[r]_{\rDR{L,a}} & \rDR{L} \ar@{-->}[u]_\rP \ar[ur]_{\rF} & 
}
}
\end{equation}
In fact, (a)$\Lra$(b) was shown in \autoref{thm:na_nmu_char}(1), and (b)$\Lra$(c)  follows from the equivalence between $\JSLf$ and $\Dep$. To see this, note that in the left diagram we may assume $q$ to be injective; otherwise, factorize $q$ as $q=q'\circ e'$ with $e$ surjective and $q'$ injective and work with $q'$ instead of $q$. By the self-duality of $\JSLf$, dualizing $q$ yields a surjective morphism from $\BLD{L} \cong \BLD{L}^\op$ to $S^\op$. Thus,
\[ |M(S)|=|J(S^\op)| \leq |J(\BLD{L})| = |\Sigma^*/{\sim_L}| = |\LD{\rev{L}}| \leq |B|.\]
In the two last steps, we use that the congruence classes of $\sim_L$ correspond bijectively to left derivatives of $\rev{L}$ by \eqref{eq:nerodecong}, and that $\LD{\rev{L}}$ is the set of states of the minimal dfa for $\rev{L}$.

By \autoref{ex:sld-vs-dlr} and \ref{ex:bld-vs-nerode} the upper and lower path of the left diagram in $\JSLf$ correspond under the equivalence functor $\Pirr$ to the upper and lower path of the right diagram in $\Dep$. Therefore, \autoref{thm:jsl_vs_dep} shows the two diagrams to be equivalent.
\smallskip
\item From (a)$\Lra$(c) we deduce that the relations $\rS$, $\rP$, $\rQ$ and $\rT_a$ ($a\in \Sigma$) constitute a short certificate for the existence of an atomic nfa for $L$ with at most $k$ states. Commutativity of the right diagram can be checked in polynomial time because all the relations appearing in the upper and lower path can be efficiently computed from the given dfas $A$ and $B$. Indeed, for the lower path we have already noted this in \autoref{rem:dfapair}, and the upper path emerges from the minimal dfa for $\rev{L}$, using that $\Sigma^*/{\sim_L}\cong \LD{\rev{L}}$. \qedhere
\end{enumerate}
\end{proof}

\begin{rem}
An alternative proof that $\DFA+\rev{\DFA}\to \NFA_{\mathbf{atm}}$ is in $\NP$ uses the following characterization of atomic nfas. Given an nfa $N$, let $\rsc{\rev{N}}$ denote the dfa obtained by determinizing the reverse nfa $\rev{N}$ via the subset construction and restricting to its reachable part. Then $N$ is atomic iff $\rsc{\rev{N}}$ is a minimal dfa~\cite[Corollary 2]{TheoryOfAtomataBrzTamm2014}. Thus, given a pair $(A,B)$ of mutually reversed dfas, to decide whether $\nalpha{L(A)}\leq k$ one may guess a $k$-state nfa $N$ and verify that $\rsc{\rev{N}}$ is a minimal dfa equivalent to $B$. One advantage of our above categorical argument is that it yields simple certificates in the form of $\Dep$-morphisms subject to certain commutative diagrams, which amount to solutions of equations in $\Rel$. The latter may be directly computed using a SAT solver, leading to a practical approach to finding small atomic acceptors (cf.~\cite{NfaSat}). To this effect, let us note that the  proof of \autoref{prop:atomic_np} actually shows how to \emph{construct} small atomic nfas rather than just deciding their existence: every certificate $\rS, \rP, \rQ, \rT_a$ ($a\in \Sigma$) yields an atomic nfa with states $\rS_\src$, transitions given by $(\rT_a)_-\seq \rS_\src \times \rS_\src$ for $a\in \Sigma$, initial states  $(\rI\fatsemi \rP)_-[\ast]\seq \rS_\src$ and final states $(Q\fatsemi \rF')_-\spbreve[\ast]\seq \rS_\src$. (Recall that $\fatsemi$ denotes composition in $\Dep$ and $(\dash)_-$ denotes the maximum lower witness of a $\Dep$-morphism, see~\autoref{rem:dep-witnesses}.) In fact, this is precisely the nfa of join-irreducibles of the $\JSL$-dfa $(S,\tau,p\circ i, f'\circ q)$ induced by the left diagram in \eqref{eq:equiv-diagrams}. Analogous reasoning also applies to the computation of small subatomic nfas treated in \autoref{thm:subatomic_npc} below.
\end{rem}

\begin{proposition}\label{prop:atomic_nphard}
The  problem $\DFA + \rev{\DFA} \to \NFA_{\mathbf{atm}}$ is $\NP$-hard.
\end{proposition}

\begin{proof}
We devise a polynomial-time reduction from the $\NP$-complete problem \textbf{BICLIQUE COVER}~\cite{garey1979computers}: given a pair $(\rR,k)$ of a relation $\rR \seq \rR_\src\times \rR_\trg$ between finite sets and a natural number $k$, decide whether $\rR$ has a biclique cover of size at most $k$, i.e.\ $\dim{\rR}\leq k$.

  For any $(\rR$,$k$), let $S=\Open(\rR)$ be the finite semilattice of open sets corresponding to the $\Dep$-object $\rR$, cf.\ \autoref{thm:jsl_vs_dep}, and let $L=L(S)$ be its lattice language. We claim that the desired reduction is given by \[(\rR,k)\quad\longmapsto\quad (\dfa{L},\dfa{\rev{L}},k),\]
where $\dfa{L}$ and $\dfa{\rev{L}}$ are the  minimal dfas for $L$ and $\rev{L}$. Thus,
we need to prove that (a) $\dim{\rR}=\nalpha{L}$, and (b) the two dfas can be computed in polynomial time from $\R$.

\medskip\noindent 
Ad (a). We have the following sequence of $\Dep$-isomorphisms:
\[ \rR \xleftrightarrow[\text{Thm \ref{thm:jsl_vs_dep}}]{\cong} \Pirr(\Open(\rR)) = \Pirr(S) \xleftrightarrow[\text{Lem \ref{lem:lattice_lang_nuclear}}]{\cong} \Pirr(\SLD{L(S)}) = \Pirr(\SLD{L}) \xleftrightarrow[\text{Ex \ref{ex:sld-vs-dlr}}]{\cong} \rDR{L}.  \]
\autoref{lem:dep-iso-preserves-dim}(2) and \autoref{prop:lattice-complexity} then imply $\dim{\rR}=\dim{\rDR{L}}=\nalpha{L}$.

\medskip\noindent Ad (b). Let $J(\Open(\rR)) = \{ j_1, \dots, j_n \}$ and $M(\Open(\rR)) = \{ m_1, \dots, m_p \}$. Then $\dfa{L}$ and $\dfa{\rev{L}}$ are the automata depicted below, where $L$ and $\rev{L}$ are their respective initial states.
  \[
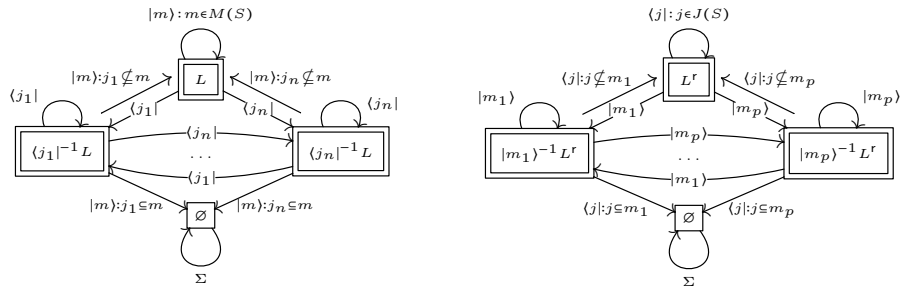
    \begin{tabular}{lll}
      \tiny
      \xymatrix@=10pt{
        && *++[F=]{L} \ar@(ul,ur)^{\ket{m} \, : \, m \in M(S)} \ar[dll]|-{\bra{j_1}} \ar[drr]|-{\bra{j_n}} &&
        \\ *++[F=]{\bra{j_1}^{-1} L} \ar@(ul,ur)^<<<{\bra{j_1}} \ar[drr]_-{\ket{m} : j_1 \subseteq m} \ar@<6pt>[urr]^-{\ket{m} : j_1 \nsubseteq m} \ar@/^7pt/[rrrr]|{\bra{j_n}} && \dots && *++[F=]{\bra{j_n}^{-1} L} \ar[dll]^-{\ket{m} : j_n \subseteq m} \ar@<-6pt>[ull]_-{\ket{m} : j_n \nsubseteq m} \ar@(ul,ur)^>>>{\bra{j_n}} \ar@/^10pt/[llll]|{\bra{j_1}}
        \\ && *+[F-]{\emptyset} \ar@(dr,dl)^{\Sigma}
      }
      &&
      \tiny
      \xymatrix@=10pt{
        && *++[F=]{\rev{L}} \ar@(ul,ur)^{\bra{j} \, : \, j \in J(S)} \ar[dll]|-{\ket{m_1}} \ar[drr]|-{\ket{m_p}} &&
        \\ *++[F=]{\ket{m_1}^{-1} \rev{L}} \ar@(ul,ur)^<<<{\ket{m_1}} \ar[drr]_-{\bra{j} : j \subseteq m_1} \ar@<6pt>[urr]^-{\bra{j} : j \nsubseteq m_1} \ar@/^7pt/[rrrr]|{\ket{m_p}} && \dots && *++[F=]{\ket{m_p}^{-1} \rev{L}} \ar[dll]^-{\bra{j} : j \subseteq m_p} \ar@<-6pt>[ull]_-{\bra{j} : j \nsubseteq m_p} \ar@(ul,ur)^>>>{\ket{m_p}} \ar@/^10pt/[llll]|{\ket{m_1}}
        \\ && *+[F-]{\emptyset} \ar@(dr,dl)^{\Sigma}
      }
    \end{tabular}
  \]
Both automata can be computed in polynomial time from $\rR$ using \autoref{lem:jirr-mirr-open}.
\end{proof}
Next, we turn to the computation of small subatomic nfas. While in the atomic case the input language was specified by a pair of dfas, we now assume an algebraic representation:

\begin{defn}
A \emph{monoid recognizer} is a triple $(M,h,F)$ of a finite monoid $M$, a map $h\colon \Sigma\to M$ and a subset $F\seq M$. The language \emph{recognized} by $(M,h,F)$ is given by $L(M,h,f)=\ol{h}^{-1}[F]$, where $\ol{h}\colon \Sigma^*\to M$ is the unique extension of $h$ to a monoid morphism.
\end{defn}
It is well-known~\cite{pin20} that a language $L$ is regular iff it has a monoid recognizer. In this case, a \emph{minimal} monoid recognizer for $L$ is given by $(\Syn{L}, \mu_L, F_L)$ where $\mu_L\colon \Sigma\to \Syn{L}$ is the domain restriction of the syntactic morphism and $F_L=\{ [w]_{\equiv_L}:w\in L \}$. It satisfies $\under{\Syn{L}}\leq \under{M}$ for every recognizer $(M,h,F)$ of $L$. Consider the following decision problem:

\medskip\noindent$\MONTOSUBATOMICNFA$\\
\textbf{Input:} A monoid recognizer $(M,h,F)$ and a natural number $k$.\\
\textbf{Task:} Decide whether there exists a $k$-state subatomic nfa accepting $L(M,h,F)$.

\medskip\noindent Here we assume that the monoid $M$ is explicitly given by its multiplication table.

\begin{theorem}\label{thm:subatomic_npc}
The problem $\MONTOSUBATOMICNFA$ is $\NP$-complete.
\end{theorem}

\begin{proof}[Proof sketch]
The proof is conceptually similar to the one of \autoref{thm:atomic_npc}. To show the problem to be in $\NP$, one uses the algebraic characterization of $\nmu{L}$ in \autoref{thm:na_nmu_char}(2) and translates the ensuing $\JSLf$-diagram into $\Dep$. To show $\NP$-hardness, one reduces from $\mathbf{BICLIQUE~COVER}$ via \[(\rR,k)\qquad\mapsto\qquad ((\Syn{L},\mu_L,F_L),k),\] where again $L=L(\Open(\rR))$.
\end{proof}

 Our complexity results indicate a trade-off, i.e.\ computing small subatomic nfas requires a less succinct representation of the input language. Generally, $|\dfa{L}|, |\dfa{L^r}| \leq |\Syn{L}|$ and the syntactic monoid can be far larger -- even for nuclear languages.

\begin{example}
For any natural number $n$ consider the dfa $A_n = (\{ 0, \dots, n-1 \}, \delta, 1, \{ 1 \})$ over the alphabet $\Sigma = \{ \pi, \tau \}$ with $\delta_\pi(i)=i+1 \mod n$ for $i=0,\cdots n-1$, and $\delta_\tau(0)=1$, $\delta_\tau(1)=0$, $\delta_\tau(i)=i$ otherwise. Let $L_n=L(A_n)$ denote its accepted language. Then:
  \begin{enumerate}
    \item Both $A_n$ and its reverse nfa are minimal dfas; in particular, $|\dfa{L_n}|=|\dfa{\rev{L}_n}|=n$.
    \item We have $|\Syn{L_n}|=n!$. To see this, recall that $\Syn{L_n}$ is the transition monoid of $A_n\cong \dfa{L_n}$. It is generated by the $n$-cycle $\delta_\pi = (0\;1\;\cdots\; n-1)$ and the transposition $\delta_\tau=(0\; 1)$; then it equals the symmetric group $S_n$ on $n$ letters.
    \item By part (1) the language ${L_n}$ is \emph{bideterministic}~\cite{TAMM2004135}, i.e.\ accepted by a dfa whose reverse nfa is deterministic. This implies that the left derivatives of $L_n$ are pairwise disjoint, so $\SLD{L_n}$ is a boolean algebra. In particular, $L_n$ is a nuclear language.
  \end{enumerate}
\end{example}
We finally further justify the inputs of $\DFA+\rev{\DFA}\to \NFA_{\mathbf{atm}}$ and $\MONTOSUBATOMICNFA$: the two modified problems $\DFA\to \NFA_{\mathbf{atm}}$ and $\DFA\to \NFA_{\mathbf{syn}}$ where only a (single) dfa is given are computationally much harder.

\begin{theorem}\label{thm:atomic-subatomic-pspacecomplete}
$\DFA\to \NFA_{\mathbf{atm}}$ and $\DFA\to \NFA_{\mathbf{syn}}$ are $\PSPACE$-complete.
\end{theorem} 

\begin{proof}
This follows by inspecting Jiang and Ravikumar's~\cite{MinNfaHard1993} argument that  $\DFA\to\NFA$ is $\PSPACE$-complete. These authors give a polynomial-time reduction from the $\PSPACE$-complete problem \textbf{UNIVERSALITY OF MULTIPLE DFAS}, which asks whether a given list $A_1,\ldots, A_n$ of dfas over the same alphabet $\Sigma$ satisfies $\bigcup_i L(A_i)=\Sigma^*$. For any $A_1,\ldots,A_n$ they construct a dfa $A$ over some alphabet $\Gamma$ and a natural number $k$ such that:
\begin{enumerate}
\item If $\bigcup_i L(A_i)\neq \Sigma^*$, then every nfa accepting $L(A)$ requires at least $k+1$ states.
\item If $\bigcup_i L(A_i)= \Sigma^*$, then there exists an nfa accepting $L(A)$ with $k$ states.
\end{enumerate} 
In the proof of (2), an explicit $k$-state nfa $N=(Q,\delta,\{q_0\},F)$ with $L(N)=L(A)$ is given, see \cite[Fig.~1]{MinNfaHard1993}. It has the property that, after $\epsilon$-elimination, for every state $q$ there exists $w\in \Gamma^*$ with $\delta_w[q_0]=\{q\}$. This implies that every state $q$ accepts a left derivative $w^{-1}L(N)$, i.e.~$N$ is a \emph{residual} nfa~\cite{ResidFSA2001}. In particular, $N$ is both atomic and subatomic. Consequently,
$(A_1,\ldots, A_n)\mapsto (A,k)$ is also a reduction to both $\DFA\to \NFA_{\mathbf{atm}}$ and $\DFA\to \NFAmu$.
\end{proof}

\section{Applications}\label{sec:applications}
We conclude this paper by outlining some useful consequences of our  $\NP$-completeness results concerning the computation of small nfas for specific classes of regular languages. 

\subsection{Nuclear Languages}
As shown above, nuclear languages form a natural common generalization of bideterministic, biRFSA, and lattice languages. Let $\DFA+\rev{\DFA}\to \NFA$ be the variant of $\DFA+\rev{\DFA}\to \NFA_{\mathbf{atm}}$ where the target nfas are arbitrary, i.e.\ the task is to decide $\ns{L(A)}\leq k$. Then:
\begin{theorem}\label{thm:nuclear_n_complete}
For nuclear languages, the problem $\DFA + \rev{\DFA} \to \NFA$ is $\NP$-complete.
\end{theorem}
In fact, by \autoref{prop:nuclear-complexity}(1) we have $\ns{L}=\dim{\rDR{L}}$ for nuclear languages, so $\NP$ certificates are given by biclique covers. The $\NP$-hardness proof is identical to the one of \autoref{thm:atomic_npc}: the reduction involves a lattice language, which is nuclear by \autoref{lem:lattice_lang_nuclear}.

\subsection{Unary languages}
For unary regular languages $L\seq \{a\}^*$, every two-sided derivative $(a^i)^{-1}L(a^j)^{-1}$ is equal to the left derivative $(a^{i+j})^{-1}L$. Therefore, we have $\nalpha{L}=\nmu{L}$ and the minimal dfa for $L$ is the dfa structure of the syntactic monoid. From \autoref{thm:subatomic_npc} we thus derive

\begin{theorem}
  \label{thm:unary_subatomic_npc}
  For unary languages, the problem $\DFA \to \NFAmu$ is in $\NP$.
\end{theorem}
This theorem generalizes the best-known complexity result for unary nfas, which asserts that the problem $\DFA\to \NFA$ is in $\NP$ for \emph{unary cyclic languages}~\cite{MinNfaUnary91}, i.e.\ unary regular languages whose minimal dfa is a cycle. In fact, for any such language $L$ we have shown in \cite[Example 5.1]{mmu21} that $\nmu{L}=\ns{L}$, hence $\DFA\to \NFA$ coincides with $\DFA \to \NFAmu$.

\subsection{Group languages}
A regular language is called a \emph{group language} if its syntactic monoid forms a group. Several  equivalent characterizations of group languages are known; for instance, they are precisely the languages accepted by {measure-once quantum finite automata}~\cite{bro-pip1999}. Concerning their state-minimal (sub)atomic acceptors, we have the following result:

\begin{proposition}\label{prop:group-language-complexity}
For any group language $L$, we have $\nmu{L}=\nalpha{L}$.
\end{proposition}
Therefore, \autoref{thm:atomic_npc} implies

\begin{theorem}
  For group languages, $\DFA + \rev{\DFA} \to \NFAmu$ is in $\NP$.
\end{theorem}
The complexity of the general $\DFA + \rev{\DFA} \to \NFAmu$ problem is left as an open problem.

\section{Conclusion and Future Work}

Approaching from an algebraic and category-theoretic angle we have studied the complexity of computing small (sub)atomic nondeterministic machines. We proved this to be much more tractable than the general case, viz. $\NP$-complete as opposed to $\PSPACE$-complete, provided that one works with a representation of the input language by a pair of dfas or a finite monoid, respectively. There are several interesting directions for future work.

The particular form of our main two $\NP$-complete problems suggests an investigation of their variants $\DFA+\rev{\DFA}\to \NFA$ and $\mathbf{MON}\to\NFA$ computing unrestricted nfas. The reductions used in the proof of \autoref{thm:atomic_npc} and \ref{thm:subatomic_npc} show both problems to be $\NP$-hard, and we have seen in \autoref{thm:nuclear_n_complete} that they are in $\NP$ for nuclear languages. The complexity of the general case is left as an open problem.

The classical algorithm for state minimization of nfas is the Kameda-Weiner method~\cite{KamedaWeiner1970}, recently given a fresh perspective based on atoms of regular languages~\cite{tamm16}. The algorithm involves an enumeration of biclique covers of the dependency relation $\rDR{L}$. Since our base equivalence $\JSLf\simeq\Dep$ reveals a close relationship between biclique covers and semilattice morphisms (e.g.\ \autoref{lem:dep-iso-preserves-dim}), we envision a purely algebraic account of the Kameda-Weiner method. We should also compare our canonical machines to the Universal Automaton \cite{UniversalAut2008}, a language-theoretic presentation of the Kameda-Weiner algorithm. For example, our morphisms preserve the language whereas the Universal Automaton uses simulations.

Finally, the classes of nuclear and lattice languages -- introduced as technical tools for our $\NP$-completeness proofs -- deserve to be studied in their own right. For instance, we expect to uncover connections between lattice languages and the characterization of finite simple non-unital semirings which are not rings \cite[Theorem 1.7]{ZumbSemiring2007}.





\takeout{
\section{Rob's Realm}

\begin{theorem}
  If $L$ is a nuclear language then $ns(L) = n\mu(L)$.
\end{theorem}

\begin{proof}
  Previously we proved $ns(L) = na(L)$ whenever $L := L(S_0)$ is a lattice language. We extended super-semilattices $\SLD{L} \subseteq S$ to super $\JSL$-dfas $A = (S, (\delta_a)_{a \in \Sigma}, L, F)$ where:
  \[
    \delta_{\bra{j}} := \bot_S \ostar \bra{j}^{-1} L
    \qquad
    \delta_{\ket{m}} := dr_L(\bra{m}^{-1} L^r) \ostar \top_S,
  \]
  and showed $\simple{A} \subseteq \BLD{L}$. We may assume $A = \simple{A}$ for otherwise we can quotient and extend. We may enforce $\top_A = \top_{\SLD{L}} = L$ and also $F = S \setminus \{\bot_S\}$. 
  
  \begin{enumerate}
    \item 
    We begin by proving the refinement $A \subseteq \DLD{L}$.  Instead of the Myhill-Nerode left congruence we'll use the preorder $\leq_L\ \subseteq \Sigma^* \times \Sigma^*$ defined:
    \[
      v_1 \leq_L v_2 :\iff \forall X \in \LD{L}. [ v_1^r \in X \To v_2^r \in X ]
    \]
    or equivalently $v_1 \leq_L v_2 \iff (v_1^r)^{-1} L^r \subseteq (v_2^r)^{-1} L^r$. Then \emph{assuming $v_1^{-1} L^r \subseteq v_2^{-1} L^r$ and $v_1^r \in s \in S$ we must show $v_2^r \in s$}. The reverse of a lattice language is the lattice language of the order-dual lattice, modulo relabelling $\bra{j} \mapsto \ket{j}$ and $\ket{m} \mapsto \bra{m}$. There are three cases.
    \begin{enumerate}
      \item $v_1^{-1} \rev{L} = \emptyset$ i.e.\ $v_1 \in \overline{\rev{L}}$.
      \item $v_2^{-1} \rev{L} = \rev{L}$ i.e.\ $v_2 = \epsilon$ or $v_2 \in L^r \cap \Sigma^* \{ \bra{j} : j \in J(S_0) \}$.
      \item $v_i^{-1} \rev{L} = \ket{m_i}^{-1} \rev{L}$ where $m_2 \leq_{S_0} m_1$ i.e.\ $v_i = w_i \ket{m_i}$ where $w_i \in L^r$.
    \end{enumerate}
  
    Case (1) never occurs i.e.\ we cannot have $v_1^r \in \overline{L}$ because $\top_S = L$. For (2) if $v_1^r \in s$ then it is non-empty hence $\epsilon \in s$. Moreover if $v_2^r = \bra{j}w$ then $v_2^r \in s$ iff $w \in \bra{j}^{-1} L$ via $\delta_{\bra{j}}$, which holds because we also assume $v_2^r \in L$. For case (3) assume $\ket{m_1} w_1^r \in s$ and also:
    \[
      m_1 \leq_{S_0^\pOp} m_2  
      \iff \bra{m_1}^{-1} L^r \subseteq \bra{m_2}^{-1} L^r
      \iff dr_L(\bra{m_2}^{-1} L^r) \subseteq dr_L(\bra{m_1}^{-1} L^r).
    \]
    By considering $\delta_{\ket{m_1}}$ we have $s \nsubseteq dr_L(\bra{m_1}^{-1} L^r)$, hence $s \nsubseteq dr_L(\bra{m_2}^{-1} L^r)$ too. Thus $\ket{m_2} w_2^r \in s$ because $w_2^r \in L$ by assumption.

    Then we've established $A = \simple{A} \subseteq \DLD{L}$ for lattice  languages $L$.

    \item
    We now show $\gamma_{U_1} = \gamma_{U_2} \iff \delta_{U_1} = \delta_{U_2}$ for all $U_1, U_2 \subseteq \Sigma^*$ where $\gamma_a : \SLD{L} \to \SLD{L}$ are the transitions of a lattice language $L$.
    Letting $M(v) := dr_L(v^{-1} L^r)$ we first establish:
    \[
      M(v_0) \ostar u_0^{-1} L \leq \gamma_U
      \iff M(v_0) \ostar u_0^{-1} L \leq \delta_U
      \qquad
      \text{for any $u_0, v_0 \in \Sigma^*$ and $U \subseteq \Sigma^*$}.
    \]
    The implication $(\oT)$ follows because $\delta_U$ restricts to $\gamma_U$ and $M(v_0) \ostar u_0^{-1} L$ restricts too. Conversely $M(v_0) \ostar u_0^{-1} L \leq \gamma_U$ means $\forall u \in \Sigma^*. [ v_0^r \in u^{-1} L \To u_0^{-1} L \subseteq U^{-1}(u^{-1} L) ]$. We know each join-irreducible of $S$ is a union of intersections $j_0 = \bigcup_i \bigcap_k u_{ik}^{-1} L$. Then whenever $v_0^r \in j_0$ we know $\exists i_0. v_0^r \in \bigcap_k u_{i_0 k}^{-1} L$ and thus $u_0^{-1} L \subseteq U^{-1} \bigcap_k u_{i_0 k}^{-1} L \subseteq U^{-1} j_0$. So we have established $M(v_0) \ostar u_0^{-1} L \leq \delta_U$ as desired.

    So assume $\gamma_{U_1} = \gamma_{U_2}$ where $U_1, U_2 \subseteq \Sigma^*$. We know $M(v) \ostar u^{-1} L \leq \gamma_{U_i} \iff M(v) \ostar u^{-1} L \leq \delta_{U_i}$ for $i = 1,2$. Moreover $\delta_{\bra{j}} = \Lor_{v \in \Sigma^*} M(v) \ostar \bra{j}^{-1} L$ because $s \nsubseteq M(v) \iff v^r \in s$, and $\delta_{\ket{m}} = \Lor_{u \in \Sigma^*} dr_L(\bra{m}^{-1} L^r) \ostar u^{-1} L$ because $\top_S = L = [\Sigma^*]^{-1} L$. Since composition is bilinear w.r.t.\ joins, $\delta_{U_i}$ is a join of $M(v) \ostar u^{-1} L$'s for $i = 1,2$. Thus $\delta_{U_1} = \delta_{U_2}$ as desired.

    \item Finally suppose $L$ is a nuclear language. Then we have the lattice language $L_0 := L(\SLD{L})$. Given $\SLD{L} \subseteq S$ then since $\SLD{L_0} \cong \SLD{L}$

  \end{enumerate}

\end{proof}
}

\takeout{
\subsection{Missing bits and bobs}

\begin{lemma}
  \label{lem:up_morphism_basics}
  \begin{enumerate}
    \item Given $s \ostar t : S \to T$ then $(s \ostar t)_* = t \ostar s : T^{\pOp} \to S^{\pOp}$.
    \item If $s \ostar t_1 : S \to T$ and $t_2 \ostar u : T \to U$ then:
    \[
      t_2 \ostar u \circ s \ostar t_1
      = \begin{cases}
        s \ostar u & \text{if $t_1 \nleq_T t_2$}
        \\
        \top_S \ostar \bot_U & \text{otherwise}.
      \end{cases}
    \]
  \end{enumerate}
\end{lemma}

\begin{proof}
  For the first claim observe $s_0 \ostar_{S, T} t_0 (s) \leq_T t \iff s \leq_S t_0 \ostar_{T^\pOp, S^\pOp} s_0 (t)$ because $s \leq_S s_0 \,\lor\, t_0 \leq_T t$ holds iff $t \leq_{T^\pOp} t_0 \,\lor\, s \leq_S s_0$ holds. The second claim is another routine calculation, noting that $\top_S \ostar \bot_U = \lambda s. \bot_U$.
\end{proof}

The morphisms $\JSL_f(S, T)$ define a semilattice via \emph{pointwise-joins}.

\begin{lemma}
  If $f : S \to T$ is a morphism between finite semilattices then t.f.a.e.
  \begin{enumerate}
    \item $f$ is nuclear.
    \item $f$ factors through a finite boolean algebra.
    \item $f$ is a pointwise-join of morphisms $s \ostar t : S \to T$ where $(s, t) \in S \times T$.
    \item $f$ is a pointwise-join of morphisms $m \ostar j : S \to T$ where $(m, j) \in M(S) \times J(T)$.
  \end{enumerate}
\end{lemma}

\begin{proof}
  \begin{enumerate}
    \item[($1 \iff 2$)]
    Assume $f$ factors through a finite distributive lattice $D$. It suffices to show $D$ arises as a retract $r \circ e = id_D$ of a finite boolean algebra. Recalling that join-irreducibles of finite distributive lattices are join-prime \cite{GratzerGeneralLattice1998}, we have the well-defined retract:
    \[
      \begin{tabular}{lllll}
        $e : D \monoto \Pow J(D)$
        & $r : \Pow J(D) \epito D$
        &&
        $e(a) := \{ j \in J(D) : j \leq_D q \}$
        & $r(S) := \Lor_D S$.
      \end{tabular}
    \]

    \item[($2 \iff 3$)]
      We now know $f : S \to T$ is nuclear iff $f = S \xto{\alpha} 2^Z \xto{\beta} T$ for some finite set $Z$ and morphisms $\alpha$, $\beta$. The coproduct and product coincide in $\JSL_f$ so we equivalently have morphisms $(\alpha_z : S \to \mathbb{2})_{z \in Z}$ and $(\beta_z : \mathbb{2} \to T)_{z \in Z}$ such that:
      \[
        f = [\beta_z]_{z \in Z} \circ \ang{\alpha_z}_{z \in Z}
        \qquad
        \text{where}\quad
        \begin{cases}
          \ang{\alpha_z}_{z \in Z}(s)
          &
          := \lambda z \in Z. \alpha_z(s)
          \\
          [\beta_z]_{z \in Z}(\delta : Z \to \mathbb{2})
          &
          := \Lor_T \{ \beta_z(\delta(z)) : z\in Z \}.
        \end{cases}
      \]
      Each morphism $\beta_z : \mathbb{2} \to T$ satisfies $\beta_z = \bot_S \ostar b_z$ where $b_z := \beta_z(1)$; dually each $\alpha_z : S \to \mathbb{2}$ satisfies $\alpha_z = a_z \ostar \top_S$ where $a_z := (\alpha_z)_*(0)$. Thus:
      \[
        \begin{tabular}{lll}
        $f$
        & $= \Lor_{\JSL(S,T)} \{ (\beta_z \circ \alpha_z) : z \in Z \}$
        & see above
        \\&
        $= \Lor \{  ((\bot_S \ostar b_z )\circ (a_z \ostar \top_S))  : z \in Z \}$
        & see above
        \\ &
        $= \Lor \{ a_z \ostar b_z : z \in Z \}$
        & by Lemma \ref{lem:up_morphism_basics}, assuming $|S| > 1$.
        \end{tabular}
      \]
      It follows that nuclear morphisms $f : S \to T$ are precisely the joins of morphisms $s \ostar t$. 

    \item[($3 \iff 4$)]
    We'll show:
    \[
      s_0 \ostar t_0 =
      \Lor_{\JSL(S, T)} \{ m \ostar j : s_0 \leq_S m \in M(S), \, J(T) \ni j \leq_T t_0 \}
      \qquad
      \text{for any $(s_0, t_0) \in S \times T$}.
    \]
    Each summand satisfies $m \ostar j \leq s_0 \otimes t_0$ because $s \nleq_S m$ implies $s \nleq_S s_0$ and moreover $j \leq_T t_0$. Conversely, whenever $s \nleq_S s_0$ then $\exists m \in M(S).[ s_0 \leq_S m \,\land\, m \nleq_S s_0]$ because $s$ is a meet of meet-irreducibles, so the action of RHS on $s$ is no less than $\Lor_T \{ j \in J(T) : j \leq_T t_0 \} = t_0$.
  \end{enumerate}
\end{proof}

We can also take joins of $\Dep$-morphisms i.e.\ \emph{unions}. In particular $\rR_1 \cup \rR_2 : \rG \to \rH$ is a morphism by taking the union of the respective witnesses, and yields the pointwise-join $\lambda Y. (\rR_1)_+\spbreve[Y] \cup (\rR_2)_+\spbreve[Y]$. To put it another way, $\JSL_f \cong \Dep$ restricts to lattice isomorphisms between hom-semilattices. We now apply this reasoning to nuclear morphisms.

\begin{lemma}
  \begin{enumerate}
    \item $\Pirr\, m \ostar j = (\Pirr S)[m] \times (\Pirr T)\spbreve[j]$ for any $(m, j) \in M(S) \times J(T)$.
    \item $\Pirr\, s \ostar t = \bigcup \{ (\Pirr S)[m] \times (\Pirr T)\spbreve[j] : s \leq_S m, j \leq_T t \}$ for any $(s, t) \in S \times T$.
  \end{enumerate}
\end{lemma}

\begin{proof}
  $\Pirr\, m_0 \ostar j_0 (j, m)$ iff $m_0 \ostar j_0(j) \nleq_T m$ iff $j \nleq_S m_0 \,\land\, j_0 \nleq_T m$ because $\bot_T \leq_T m_0$ i.e.\ the first claim holds. The second follows via Lemma \ref{lem:nuclear_char} because the pointwise-join of $\Dep$-morphisms amounts to taking their union.
\end{proof}

\begin{definition}
  $\DLD{L}$ is the closure of $\LD{L}$ under unions and intersections.
\end{definition}

\begin{lemma}
  $M(\DLD{L}) = dr_L[\LD{L}]$.
\end{lemma}

\begin{proof}
  By the Dependency theorem $M_v := dr_L(v^{-1} L^r) = \bigcup \{ X \in \LD{L} : v^r \nin X \}$, so that $\overline{M_v} = \bigcap \{ X \in \LD{\overline{L}} : v^r \in X \}$. The $\overline{M_v}$ clearly join-generate $\DLD{\overline{L}}$ hence $dr_L[\LD{L}]$ meet-generates $\DLD{L} \cong (\DLD{\overline{L}})^\pOp$. Finally each $dr_L(v^{-1} L^r)$ is meet-irreducible because it is covered by the smallest $X \supseteq dr_L(v^{-1} L^r)$ such that $v^r \in X$.
\end{proof}
}

\takeout{
\subsection{Standalone bits}

\begin{theorem}
  $\DFA \to \NFA_{\mathbf{syn}}$ is PSPACE-complete.
\end{theorem}

\begin{proof}
 Jiang and Ravikumar proved $\DFA \to \NFA$ is PSPACE-complete in \cite[Theorem 3.2]{jr91}. Modulo dfa minimization, the input is $(\dfa{L}, k)$ and the question is $\ns{L} \leq k$? The crucial PSPACE-hardness reduction uses the \emph{universality of the union of $n$ dfas} problem, a reformulation of \cite[Lemma 3.2.3]{Kozen1977}. Their construction is rather involved so we will refer to the notation from their proof.
 
 Fix dfas $(M_i)_{1 \leq i \leq n}$ over common alphabet $\Sigma$. They define $k := 4 + \sum_{1 \leq i \leq n} |M_i|$ and (implicitly) construct a small dfa $D$ accepting $L := L(D)$ over alphabet $\Sigma_0 \supseteq \Sigma$ such that:
 \[
  k \leq ns(L) \leq k + 1
  \qquad
  ns(L) = k \iff \bigcup_{1\leq i \leq n} L(M_i) = \Sigma^*.
 \]
 The $k+1$ upper bound is witnessed by an explicit nfa denoted $N'$.\footnote{Technically their $N'$ has $\epsilon$-transitions but we may assume they have been eliminated} It is constructed from the disjoint union $\coprod_{1 \leq i \leq n} M_i$ by adding $5$ special states $\{ q_0, f, p_1, p_2, p_3 \}$ and various transitions labelled using $\Sigma_0$.
 In particular, $L(N') = L$. In the crucial situation $\bigcup_{1\leq i \leq n} L(M_i) = \Sigma^*$ they observe $L(N'')= L$ where $N''$ is $N'$ without the state \emph{f}. Importantly they prove $k \leq ns(L)$ regardless of the answer to the original decision problem,  essentially by the biclique edge cover method \cite[Theorem 4]{LowerBoundsHard}. Finally they prove $ns(L) = k \iff \bigcup_{1\leq i \leq n} L(M_i) = \Sigma^*$ which (recalling the witness $N''$) amounts to proving $ns(L) \geq k + 1$ whenever $\bigcup_{1\leq i \leq n} L(M_i) \neq \Sigma^*$.
 
 Having summarised their approach, we observe properties of the nfa $N'' = (X, N_a, I, F)$.
  \begin{enumerate}
    \item $q_0$ occurs in precisely one reachable subset, namely $I$.
    \item We have $N''_c[I] = \{ p_2 \}$, $N''_{cd}[I] = \{ p_3 \}$, $N''_{cdf}[I] = \{ p_1 \}$.
    \item Each state of each dfa $M_i$ occurs in a singleton reachable subset $N''_{a_i w}[I]$ where $w \in \Sigma$.
  \end{enumerate}
  
  Consider the $\JSL$-dfa $\delta := \jslReach{\Pow N''}$ obtained by closing the reachable subset construction $rsc(N'')$ under unions. By the above remarks, the underlying semilattice of $\delta$ has precisely $|N''| = k$ join-irreducibles. Let us assume $ns(L) = k$ or equivalently $L(N'') = L$. Then the behaviour map provides a surjective semilattice morphism $\beh : \delta \epito \SLD{L}$, hence $|J(\SLD{L})| \leq k$. But since the respective nfa of irreducibles $N_{L}$ accepts $L$ and has $|J(\SLD{L})|$ states, $N_{L}$ has $k$ states by the state-minimality of $N''$. Thus $ns(L) = k \iff |N_L| = k \iff n\mu(L) = k$ because $N_L$ is subatomic.
\end{proof}

{\bf TODO}

\begin{remark}[Simple semirings]
  The syntactic monoids of lattice languages are intimately related to the finite \emph{simple} semirings with zero \cite[Theorem 1.7]{ZumbSemiring2007}.
  Let $R$ be a finite semiring with zero; we do not require a unit and may assume its multiplication is non-trivial.\footnote{A semiring with zero has trivial multiplication if $x \cdot y = 0$. The respective finite simple semirings correspond to simple abelian groups and the two element semilattice.}

  \begin{quote}
    Modulo isomorphism, $R$ is simple iff (a) $R = End(V)$ for some finite vector space $V$, or (b) $\ts{\SLD{L}} \subseteq R \subseteq End(\SLD{L})$ for some lattice language $L$.
  \end{quote}

  The characterisation of finite simple rings comes from the Artin-Wedderburn theorem. The other cases amount to \cite[Theorem 1.7]{ZumbSemiring2007}, where $\ts{\SLD{L}}$ is the non-unital transition semiring of the $\JSL$-dfa $\SLD{L}$ (see \cite{mmu21}). If $L$ is a lattice language these semirings are precisely the \emph{dense} cores used by Zumbr\"{a}gel. They are also  non-unital syntactic semirings.
\end{remark}

\subsection{$\Dep$ is equivalent to $\JSL_f$}

We provide a self-contained proof of the equivalence $\Dep \cong \JSL_f$, originally proved in \cite{myers2020representing}. We'll describe witnessing functors $\Open : \Dep \to \JSL_f$, $\Nleq : \JSL_f \to \Dep$ and also $\Pirr : \JSL_f \to \Dep$ i.e.\ the functorial generalisation of Markowsky's poset of irreducibles \cite{MarkowskyLat1975}.

\begin{lemma}[Maximal witnesses]
  $\rR \subseteq \rG_\src \times \rH_\trg$ defines a $\Dep$-morphism $\rR : \rG \to \rH$ iff:
  \[
    \rR_- ; \rH = \rR = \rG ; \rR_+\spbreve
    \qquad\text{where}\quad
    \begin{cases}
      \rR_-(g_s, h_s) ~\mathrel{\vcentcolon\Longleftrightarrow}\ \rH[h_s] \subseteq \rR[g_s]
      \\
      \rR_+(g_t, h_t) ~\mathrel{\vcentcolon\Longleftrightarrow}\ \breve{\rG}[g_t] \subseteq \breve{\rR}[h_t].
    \end{cases}
  \]
  Moreover if $\rR_l ; \rH = \rR = \rG ; \rR_u\spbreve$ then $\rR_l \subseteq \rR_-$ and $\rR_u \subseteq \rR_+$.
\end{lemma}

\begin{proof}
  If $\rR_- ; \rH = \rR = \rG ; \rR_+\spbreve$ we have a $\Dep$-morphism $\rR : \rG \to \rH$. Conversely given $\rR : \rG \to \rH$ then $\rR_-$ is the maximal solution of $\rI ; \rH \subseteq \rR$, thus $\rR_- ; \rH = \rR$ because a solution attaining equality exists. Similarly $\rR_+$ is the maximal solution of $\rG ; \rI\spbreve \subseteq \rR$ so $\rG ; \rR_+\spbreve = \rR$.
\end{proof}

\begin{definition}
  \begin{enumerate}
    \item 
    The functor $\Nleq : \JSL_f \to \Dep$ is defined $\Nleq S :=\ \nleq_S\ \subseteq S \times S$ and $\Nleq f(x, y) ~\mathrel{\vcentcolon\Longleftrightarrow}\ f(x) \nleq_S y$ for any semilattice morphism $f : S_1 \to S_2$. Moreover:
    \[
      (\Nleq f)_-(s_1, s_2) \iff s_2 \leq_{S_2} f(s_1)
      \qquad
      (\Nleq f)_+(s_2, s_1) \iff f_*(s_2) \leq_{S_1} s_1.
    \]
    \item
    $\Open : \Dep \to \JSL_f$ is defined $\Open \rG := (O(\rG),\cup,\emptyset)$ where $O(\rG) := \{ \rG[S] : S \subseteq \rG_\src \}$ and also $\Open \rR := \lambda Y.\rR_+\spbreve[Y]$.
    Equivalently $\Open\rR = \lambda Y.\rR_u\spbreve[Y]$ for any upper witness $\rR_u$.

    \item For any relation $\rG$ define $\inte_\rG : \Pow \rG_\trg \to \Pow \rG_\trg$ as $\inte_\rG(Y) := \bigcup \{ \rG[X] : X \subseteq \rG_\src, \rG[X] \subseteq Y \}$.
  \end{enumerate}
\end{definition}

\begin{lemma}
  $\Nleq$ and $\Open$ are well-defined functors.
\end{lemma}

\begin{proof}
  Consider $\Nleq$. Given a semilattice morphism $f : S_1 \to S_2$, the adjoint relationship $f(s_1) \leq_{S_2} s_2 \iff s_1 \leq_{S_1} f_*(s_2)$ has contrapositive form $f ; \nleq_{S_2}\ = \Nleq f  =\ \nleq_{S_1} ; f_*\spbreve$, hence $\Nleq f :\ \nleq_{S_1} \to\ \nleq_{S_2}$ is a well-defined $\Dep$-morphism. Then $\Nleq\ id_S =\ \nleq_S\ = id_{\Nleq S}$ and moreover $\Nleq (g \circ f) = \Nleq g \fatsemi \Nleq f$ since $(f ; g)_*\spbreve = (g_* ; f_*)\spbreve = f_*\spbreve ; g_*\spbreve$ by $\JSL_f$ self-duality. The maximal witnesses follow by unwinding the definitions.
  
  Consider $\Open$. Given $\rR : \rG \to \rH$ then $\Open\rR$ is a well-defined function because $\Open\rR(\rG[S]) = \rR_+\spbreve[\rG[S]] = \rH[\rR_-[S]]$ and clearly preserves unions. Moreover $\Open id_\rG = id_{\Open S}$ via upper witness $\rG = \rG ; id_{\rG_\trg}\spbreve$. Finally $\Open (\rR \fatsemi \rS) = \Open \rS \circ \Open\rR$ because $\rS_+ ; \rR_+$ is an upper witness for $\rR \fatsemi \rS$.
\end{proof}

\begin{lemma}
  \label{lem:el_interior_rep}
  $Y \nsubseteq \inte_\rG(\overline{g_t}) \iff g_t \in Y$ for any $g_t \in \rG_\trg$ and $Y \in O(\rG)$.
\end{lemma}

\begin{proof}
  If $g_t \in Y$ then $Y \nsubseteq \inte_\rG(\overline{g_t})$.
  Conversely if $g_t \nin Y$ then $Y \subseteq \overline{g_t}$ hence $Y \subseteq \inte_\rG(\overline{g_t})$.
\end{proof}

\begin{theorem}[Natural isomorphisms for $\Open$ and $\Nleq$]
  \label{thm:nat_isos_open_nleq}
  \item
  \begin{enumerate}
    \item $\alpha : Id \Rightarrow \Open\circ\Nleq$ defined $\alpha_S := \lambda x.\{ y \in S : x \nleq_S y \}$ is a natural isomorphism.
    \item $\beta : Id \Rightarrow \Nleq\circ\Open$ defined $\beta_\rG(x, Y) ~\mathrel{\vcentcolon\Longleftrightarrow} \rG[x] \nsubseteq Y$ is a natural isomorphism.
  \end{enumerate}
\end{theorem}

\begin{proof}
  Consider $\alpha$.
  Each partial order $(P, \leq_P)$ embeds into its upsets ordered by reverse-inclusion via $p \mapsto\ \up_{\leq_P} p$. Since $P$-joins becomes intersections, $\alpha$ is an injective semilattice morphism. It is surjective because $\nleq_S[X] =\ \nleq_S[\Lor_S X]$ for any $X \subseteq S$, hence it is an isomorphism. Concerning naturality,
  \[
  \begin{tabular}{lll}
    $\Open(\Nleq f) \circ \alpha_{S_1}(x_1)$
    &
    $= \Open(\Nleq f)(\{ y_1 \in S_1 : x_1 \nleq_{S_1} y_1\})$
    \\ &
    $= (\Nleq f)_+\spbreve[\{ y_1 \in S_1 : x_1 \nleq_{S_1} y_1 \}]$
    \\ &
    $= \{ y_2 \in S_2 : \exists y_1 \in S_1.(x_1 \nleq_{S_1} y_1 \text{ and } f_*(y_2) \leq_{S_1} y_1) \}$
    \\ &
    $= \{ y_2 : \neg\forall y_1 \in S_1.(f_*(y_2) \leq_{S_1} y_1 \To x_1 \leq_{S_1} y_1 ) \}$
    \\ &
    $= \{ y_2 : x_1 \nleq_{S_1} f_*(y_2) \}$
    \\ &
    $= \{ y_2 : f(x_1) \nleq_{S_2} y_2 \}$
    & via adjoints
    \\ &
    $= \alpha_{S_2} \circ f(x_1)$.
    \end{tabular}
  \]
  Consider $\beta$. Given any relation $\rG$ we have the commutative $\Rel$-diagram below.
  \[
    \xymatrix@=15pt{
      \rG_\trg \ar[rrr]^-{\nin} &&& O(\rG) \ar[rrr]^-{(\lambda x. \inte_{\rG}(\overline{x}))\spbreve} &&& \rG_\trg
      & \beta_\rG \fatsemi \breve{\in} = (\lambda x.\rG[x]) ; \breve{\in} = \rG = id_\rG
      \\
      \rG_\src \ar[u]^{\rG} \ar[urrr]^{\beta_\rG} \ar[rrr]_-{\lambda x. \rG[x]} &&& O(\rG) \ar[u]_{\nsubseteq} \ar[urrr]^{\breve{\in}} \ar[rrr]_-{\supseteq ; (\lambda x. \rG[x])\spbreve} &&& \rG_\src \ar[u]_\rG
      & \breve{\in} \fatsemi \beta_\rG = \breve{\in} ; \nin\ = \Nleq(\Open\rG) = id_{\Nleq(\Open\rG)}
    }
  \]
  The left square commutes because $\rG[x] \nsubseteq Y$ iff $\exists y. \rG(x, y) \land y \nin Y$, so $\beta_\rG$ is a $\Dep$-morphism. Concerning the right square, $y \in \rG[X]$ iff $\exists x \in X. y \in \rG[x]$ (lower triangle) iff $\rG[X] \nsubseteq \inte_\rG(\overline{y})$ (upper triangle) by Lemma \ref{lem:el_interior_rep}. Then $\breve{\in}$ is a $\Dep$-morphism, so $\beta_\rG$ is a $\Dep$-isomorphism by the compositions above right.  It remains to verify naturality -- the left $\Dep$-diagram below. Now $\Open$ is faithful i.e.\ for $\rR, \rS : \rG \to \rH$, $\rR = \rS \iff \rR_+ = \rS_+ \iff \Open\rR = \Open\rS$ since $\Dep$-morphisms are determined by the action of their maximal witnesses on $\rG[g_s]$'s.
  \[
  \xymatrix@=15pt{
    \rG \ar[rr]^-{\beta_\rG} \ar[d]_{\rR} && \Nleq(\Open\rG) \ar[d]^{\Nleq \circ \Open\rR}
    \\
    \rH \ar[rr]_-{\beta_\rH} && \Nleq(\Open\rH) 
    }
    \qquad
    \xymatrix@=15pt{
    \Open\rG \ar[rr]^-{\Open\, \beta_\rG} \ar[d]_{\Open\rR} && \Nleq(\Open\rG) \ar[d]^{\Open(\Nleq \circ \Open\rR)}
    \\
    \Open\rH \ar[rr]_-{\Open\, \beta_\rH} && \Open(\Nleq(\Open\rH))
  }
  \]
  Then it suffices to establish the commutative diagram above right.
  It instantiates $\alpha$'s naturality i.e.\ $\Open \, \beta_\rG = \lambda Y. \nin[Y] = \lambda Y. \{ Y_0 : Y \nsubseteq Y_0 \} = \alpha_{\Open\rG}$ via the earlier upper witness.
\end{proof}

\begin{corollary}
   We have the categorical equivalence $\Dep \cong \JSL_f$ via $\Open$ and $\Nleq$.
\end{corollary}

Finally we provide an alternative equivalence functor $\Pirr : \JSL_f \to \Dep$. It is the domain/codomain restriction of $\Nleq$ to join/meet-irreducibles respectively. 

\begin{definition}
  $\Pirr : \JSL_f \to \Dep$ is defined $\Pirr S :=\ \nleq_S\ \subseteq J(S) \times M(S)$ and $\Pirr f :=\ f ;\ \nleq_{S_2}\ \subseteq J(S_1) \times M(S_2)$ for any semilattice morphism $f : S_1 \to S_2$. Moreover:
  \[
    (\Pirr f)_-(j_1, j_2) \iff j_2 \leq_{S_2} f(j_1)
    \qquad
    (\Pirr f)_+(m_2, m_1) \iff f_*(m_2) \leq_{S_1} m_1.
  \]
\end{definition}

\begin{lemma}
  $\Pirr: \JSL_f \to \Dep$ is a well-defined functor.
\end{lemma}

\begin{proof}
  The witnesses $(\Nleq f)_- ; \nleq_{S_2} = \Nleq f =\ \nleq_{S_1} ; (\Nleq f)_+\spbreve$ restrict to $(\Pirr f)_- ; \Pirr S_2 = \Pirr f =\ \Pirr S_1 ; (\Pirr f)_+\spbreve$. Indeed:
  \[
    \begin{tabular}{c}
      $\exists s_2 \in S_2.[ s_2 \leq_{S_2} f(j) \land s_2 \nleq_{S_2} m ]
      \iff \exists j_2 \in J(S_2).[ j_2 \leq_{S_2} f(j) \land j_2 \nleq_{S_2} m ]$
      \\
      $\exists s_1 \in S_1.[ j \nleq_{S_1} s_1 \land f_*(m) \leq_{S_2} s_1 ]
      \iff \exists m_1 \in M(S_1).[ j \nleq_{S_1} m_1 \land f_*(m) \leq_{S_2} m_1 ]$
    \end{tabular}
  \]
  since $s_2 = \Lor_S \{ j \in J(S_2) : j \leq_S s_2 \}$ and $s_1 = \Land_S \{ m \in M(S_1) : s_1 \leq_S m \}$. Now $\Pirr$ certainly preserves identities. It preserves composition i.e.\ 
  $\Pirr S_1 ; (\Pirr f)_+\spbreve ; (\Pirr g)_+\spbreve
    = \Pirr f ; (\Pirr g)_+\spbreve
    = \Pirr (g \circ f)$ via the equivalences $\exists m \in M(S_2). [f(j_1) \nleq_{S_2} m \land g_*(m_3) \leq_{S_2} m] \iff f(j_1) \nleq_{S_2} g_*(m_3) \iff g(f(j_1)) \nleq_{S_3} m_3$,
  recalling that $M(S_2)$ meet-generates $S_2$.
\end{proof}

\begin{theorem}[Natural isomorphisms $rep$ and $red$]
  \label{thm:nat_isos_open_pirr}
  \item
  \begin{enumerate}
    \item $rep : Id \Rightarrow \Open\Pirr$ with $rep_S := \lambda x.\{ m \in M(S) : x \nleq_S m \}$ is a natural isomorphism.
    \item $red : Id \To \Pirr\Open$ with $red_\rG (x, Y)$ iff $\rG[x] \nsubseteq Y$ is a natural isomorphism.
  \end{enumerate}
  
\end{theorem}

\begin{proof}
  For $rep$ restrict the argument concerning $\alpha$ from Theorem \ref{thm:nat_isos_open_nleq}. In particular, (a)  $s \mapsto\ \up_S s \cap M(S)$ is also injective, (b) $\Pirr S [X] =\ \up_S \Lor_S X \cap M(S)$ for any $X \subseteq J(S)$.

  Likewise for $red$ restrict $\beta$ to join/meet-irreducibles. Additionally use (a) $J(\Open \rG) \subseteq \{ \rG[g_s] : g_s \in \rG_\src \}$ because $\Open\rG$ is union-generated by the singleton images, (b) $M(\Open\rG) \subseteq \{ \inte_\rG(\overline{g_t}) : g_t \in \rG_\trg \}$ because $\rG[X] = \Land_{\Open\rG} \{ \inte_\rG(\overline{g_t}) : g_t \nin \rG[X] \}$ for any $X \subseteq \rG_\src$.
\end{proof}
}

\bibliographystyle{plainurl}
\bibliography{refs,bib-2019,bib-2020}

\clearpage
\appendix 

\section*{Appendix: Omitted Proofs}
This appendix provides all proofs omitted for space reasons.

\section{Details for \autoref{sec:dep}}
\begin{notation}
Recall that the self-duality $\JSLf^\op\xra{\simeq} \JSLf$ associates to every $f\colon S\to T$ the \emph{dual morphism} $f_*\colon S^\op\to T^\op$ mapping each $t\in T$ to the $\leq_S$-largest $s\in S$ with $f(s)\leq_T t$, and that we have the adjoint relationship
\[ f(s)\leq_T t \;\iff\; s\leq_S f_*(t) \qquad\text{for all $s\in S$, $t\in T$}. \]
We shall often use it in its contrapositive form:
\[ f(s)\not\leq_T t \;\iff\; s\not\leq_S f_*(t) \qquad\text{for all $s\in S$, $t\in T$}. \]
\end{notation}

We first show that relations $\rP_l$ and $\rP_u$ witnessing that $\rP\colon \rR\to\rS$ is a $\Dep$-morphism can always be replaced by \emph{maximal witnesses}, as stated in \autoref{rem:dep-witnesses}.

\begin{lemma}[Maximal witnesses]
Let $\rR$ and $\rS$ be relations between finite sets. Then
  $\rP \subseteq \rR_\src \times \rS_\trg$ defines a $\Dep$-morphism $\rP\colon \rR \to \rS$ if and only if
  \[
    \rP_- ; \rS = \rP = \rR ; \rP_+\spbreve,
    \qquad\text{where}\quad
    \begin{cases}
      \rP_-(r_\src, s_\src) ~\mathrel{\vcentcolon\Longleftrightarrow}\ \rS[s_\src] \subseteq \rP[r_\src],
      \\
      \rP_+(s_\trg, r_\trg) ~\mathrel{\vcentcolon\Longleftrightarrow}\ \breve{\rR}[r_\trg] \subseteq \breve{\rP}[s_\trg].
    \end{cases}
  \]
  Moreover if $\rP_l ; \rS = \rP = \rR ; \rP_u\spbreve$ then $\rP_l \subseteq \rP_-$ and $\rP_u \subseteq \rP_+$.
\end{lemma}

\begin{proof}
  If $\rP_- ; \rS = \rP = \rR ; \rP_+\spbreve$ we have a $\Dep$-morphism $\rP\colon \rR \to \rS$. Conversely, given $\rP : \rR \to \rS$ then $\rP_-$ is the maximal solution of $\rI ; \rS \subseteq \rP$. Thus $\rP_- ; \rS = \rP$ because a solution attaining equality exists. Similarly, $\rP_+$ is the maximal solution of $\rR ; \rI\spbreve \subseteq \rP$, so $\rR ; \rP_+\spbreve = \rP$.
\end{proof}

\section*{Proof of \autoref{thm:jsl_vs_dep}}
To prove the equivalence, it turns out to be convenient to replace $\Pirr$ with a naturally isomorphic functor $\Nleq\colon \JSLf\xra{\simeq} \Dep$ that considers the full complemented order $\not\leq_S\leq S\times S$ of a finite semilattice rather than its restriction to $J(S)\times M(S)$:

\begin{defn}
    The functor $\Nleq\colon  \JSL_f \to \Dep$ maps $S\in \JSLf$ to the $\Dep$-object
\[\Nleq(S) :=\ \nleq_S\ \subseteq S \times S\]
and a $\JSLf$-morphism $f\colon S_1\to S_2$ to the $\Dep$-morphism
\[ \Nleq(f)\colon \Nleq(S_1)\to \Nleq(S_2),\qquad  \Nleq(f)(s_1, s_2) ~\mathrel{\vcentcolon\Longleftrightarrow}\ f(s_1) \nleq_{S_2} s_2\quad\text{for $s_1\in S_1$, $s_2\in S_2$}.\]
\end{defn}

\begin{rem}\label{rem:witnesses}
The maximal witnesses for $\Nleq(f)$ are given by 
    \[
      (\Nleq f)_-(s_1, s_2) \iff s_2 \leq_{S_2} f(s_1)
      \qquad
      (\Nleq f)_+(s_2, s_1) \iff f_*(s_2) \leq_{S_1} s_1.
    \]
\end{rem}

\begin{rem}
Recall that for any $\Dep$-morphism $\rP\colon \rR\to \rS$ the map $\Open(\rP)\colon \Open(\rR)\to \Open(\rS)$ is given by $\Open(\rR)(O)=\rP\spbreve_+[O]$. One may replace $\rP_+$ by an arbitrary upper witness $\rP_u$ of $\rP$. In fact, if $O=\rR[X]$ for $X\seq \rR_\src$, then
\[ \rP\spbreve_+[O] = \rR;\rP\spbreve_+[X] = \rP[X] = \rR;\rP\spbreve_u[X] = \rP\spbreve_u[O].  \]
\end{rem}

\begin{lemma}\label{lem:nleq-open-welldefined}
  $\Nleq\colon \JSLf\to \Dep$ and $\Open\colon \Dep\to \JSLf$ are well-defined functors.
\end{lemma}

\begin{proof}
\begin{enumerate}
\item $\Nleq$ is well-defined: Given a semilattice morphism $f\colon S_1 \to S_2$, the contrapositive adjoint relationship $f(s_1) \not\leq_{S_2} s_2 \iff s_1 \not\leq_{S_1} f_*(s_2)$ shows $f ; \nleq_{S_2}\ = \Nleq(f)  =\ \nleq_{S_1} ; f_*\spbreve$, hence $\Nleq(f)\colon \nleq_{S_1} \to\ \nleq_{S_2}$ is a $\Dep$-morphism with witnesses $f$ and $f_*$. Moreover, 
\[\Nleq(\id_S) = \nleq_S = \id_{\Nleq(S)}\quad \text{and}\quad \Nleq (g \circ f) = \Nleq(f) \fatsemi \Nleq(g)\]
where the second equality uses that $(f ; g)_*\spbreve = (g_* ; f_*)\spbreve = f_*\spbreve ; g_*\spbreve$ by the self-duality of $\JSLf$.
  
\item $\Open$ is well-defined: Given $\rP\colon \rR \to \rS$, the map $\Open(\rP)\colon \Open(\rR)\to \Open(\rS)$ is a well-defined semilattice morphism because
\[ \Open(\rP)(\rR[X]) = \rP_+\spbreve[\rR[X]] = \rS[\rP_-[X]]\]
is an open set of $\rS$ for each $X\seq \rR_\src$, and moreover
$\Open(\rP)$ clearly preserves unions. We have $\Open(\id_\rR) = \id_{\Open(\rR)}$ since
\[ \Open(\id_\rR)(\rR[X]) = \rR_+\spbreve[X] = \id_{\rR_\trg}[X] = X \qquad\text{for $X\seq \rR_\src$},  \]
and $\Open (\rP \fatsemi \rQ) = \Open \rQ \circ \Open\rP$ because 
\[ \Open(\rP\fatsemi \rQ)(\rR[X]) = [\rP\fatsemi \rQ]_+[X] = \rQ_+[\rP_+[X]] = \Open(\rQ)\circ \Open(\rP)(\rR[X]) \quad\text{for $X\seq \rR_\src$}.\]
\end{enumerate}
\end{proof}

%

\begin{theorem}\label{thm:nleq_open_equiv}
The categories $\JSLf$ and $\Dep$ are equivalent via $\Nleq$ and $\Open$.  The natural isomorphisms $\alpha\colon \Id_{\JSLf} \Rightarrow \Open\circ\Nleq$ and $\beta\colon \Id_\Dep \Rightarrow \Nleq\circ\Open$ are given by 
\[
\begin{array}{llll}
\alpha_S\colon S\xra{\cong} \Open(\Nleq(S)),&  x\mapsto \{ y\in S: x \not\leq_S y \}&&(S\in \JSLf),\\
\beta_\rR\colon \rR\xra{\cong} \Nleq(\Open(\rR)),&  
\beta_\rR(x, Y) ~\mathrel{\vcentcolon\Longleftrightarrow}\ \rR[x] \nsubseteq Y && (\rR\in \Dep).
\end{array}
\]
\end{theorem}

\begin{proof}
\begin{enumerate}
\item $\alpha$ is a natural isomorphism: We first show that $\alpha_S$ is a semilattice isomorphism. Indeed, $\alpha_S$ preserves joins because $x\vee x'\not \leq_S y$ iff $x'\not\leq_S y$ or $x'\not\leq_S y$ for any $x,x',y\in S$; it is injective because in any poset every element $x$ is determined by the elements above $x$; it is surjective because $\mathop{\not\leq_S}[X]=\alpha_S(\bigvee X)$ for every $X\seq S$. To prove naturality, we compute for every $f\colon S_1\to S_2$ in $\JSLf$ and $x_1\in S_1$:
  \begin{align*}
    \Open(\Nleq(f)) \circ \alpha_{S_1}(x_1)
    &
    = \Open(\Nleq(f))(\{ y_1 \in S_1 : x_1 \nleq_{S_1} y_1\})
    \\ &
    = (\Nleq f)_+\spbreve[\{ y_1 \in S_1 : x_1 \nleq_{S_1} y_1 \}]
    \\ &
    = \{ y_2 \in S_2 : \exists y_1 \in S_1.(x_1 \nleq_{S_1} y_1 \text{ and } f_*(y_2) \leq_{S_1} y_1) \}
    \\ &
    = \{ y_2 : \neg\forall y_1 \in S_1.(f_*(y_2) \leq_{S_1} y_1 \To x_1 \leq_{S_1} y_1 ) \}
    \\ &
    = \{ y_2 : x_1 \nleq_{S_1} f_*(y_2) \}
    \\ &
    = \{ y_2 : f(x_1) \nleq_{S_2} y_2 \}
    \\ &
    = \alpha_{S_2} \circ f(x_1).
    \end{align*}
\item  $\beta$ is a natural isomorphism: For any $\rR\in \Dep$ the $\Rel$-diagram below commutes.
  \begin{equation}\label{eq:isodiag}
    \xymatrix{
      \rR_\trg \ar[rrr]^-{\nin} &&& \Open(\rR) \ar[rrr]^-{(\lambda y. \inte_{\rR}(\rR_\trg\setminus \{y\}))\spbreve} &&& \rR_\trg
      \\
      \rR_\src \ar[u]^{\rR} \ar[urrr]^{\beta_\rR} \ar[rrr]_-{\lambda x. \rR[x]} &&& \Open(\rR) \ar[u]_{\nsubseteq} \ar[urrr]^{\breve{\in}} \ar[rrr]_-{\supseteq ; (\lambda x. \rR[x])\spbreve} &&& \rR_\src \ar[u]_\rR
    }
  \end{equation}
In fact, both triangles in the left square commute by definition of $\beta_\rR$; the lower triangle in the right square commutes trivially, and the the upper triangle by definition of $\intop_{\rR}$. This proves that $\beta_{\rR}\colon \rR\to \Nleq(\Open(\rR))$ and $\breve{\epsilon}\colon \Nleq(\Open(\rR)) \to \rR$ are $\Dep$-morphisms. Moreover, $\breve{\in}$ is the inverse of $\beta_{\rR}$ in $\Dep$: we have
\[ \beta_\rR \fatsemi \breve{\in} = (\lambda x.\rR[x]) ; \breve{\in} = \rR = \id_\rR \quad\text{and}\quad \breve{\in} \fatsemi \beta_\rR = \breve{\in} ; \nin\ = \Nleq(\Open(\rR)) = \id_{\Nleq(\Open(\rR))}.\]
Thus, $\beta_{\rR}$ is a $\Dep$-isomorphism.  It remains to verify naturality, i.e.\ that the left $\Dep$-diagram below commutes for all $\rP\colon \rR\to \rS$.
  \[
  \xymatrix@=15pt{
    \rR \ar[rr]^-{\beta_\rR} \ar[d]_{\rP} && \Nleq(\Open(\rR)) \ar[d]^{\Nleq \circ \Open(\rP)}
    \\
    \rS \ar[rr]_-{\beta_\rS} && \Nleq(\Open(\rS)) 
    }
    \qquad
    \xymatrix@=15pt{
    \Open(\rR) \ar[rr]^-{\Open(\beta_\rR)} \ar[d]_{\Open(\rP)} && \Open(\Nleq(\Open(\rR))) \ar[d]^{\Open(\Nleq \circ \Open(\rP))}
    \\
    \Open(\rS) \ar[rr]_-{\Open(\beta_\rS)} && \Open(\Nleq(\Open(\rS)))
  }
  \]
 Note that the functor $\Open$ is faithful: for any two parallel $\Dep$-morphisms $\rP$, $\rQ$ we have
\[ \rP = \rQ \iff \rP_+ = \rQ_+ \iff \Open(\rP) = \Open(\rQ).\]
Thus, it suffices to show commutativity of the right diagram. For each $Y\in \Open(\rR)$ we have
\[ \Open(\beta_\rR)(Y) = (\beta_{\rR})\spbreve_+[Y] = \not\in[Y] = \{ Y_0\in \Open(\rR) : Y\not\seq Y_0 \} = \alpha_{\rR}. 
\]
This shows $\Open(\beta_{\rR})=\alpha_{\Open(\rR)}$ and analogously $\Open(\beta_\rS) =\alpha_{\Open(\rS)}$, so the commutativity of the right diagram follows from the naturality of $\alpha$.\qedhere
\end{enumerate}
\end{proof}
Finally, we show that also $\Pirr\colon \JSL_f \to \Dep$, the restriction of $\Nleq$ to join- and meet-irreducible elements, is an equivalence functor.

\begin{rem}
In analogy to \autoref{rem:witnesses}, for any $\JSLf$-morphism $f\colon S_1\to S_2$ the maximal lower and upper witnesses for $\Pirr(f)\colon \Pirr(S_1)\to \Pirr(S_2)$ are given by
  \[
    (\Pirr f)_-(j_1, j_2) \iff j_2 \leq_{S_2} f(j_1),
    \qquad
    (\Pirr f)_+(m_2, m_1) \iff f_*(m_2) \leq_{S_1} m_1.
  \]
\end{rem}

\begin{lemma}
  $\Pirr: \JSL_f \to \Dep$ is a well-defined functor.
\end{lemma}

\begin{proof}
For any $\JSLf$-morphism $f\colon S_1\to S_2$, we have
\[(\Pirr f)_- ; \Pirr S_2 = \Pirr f =\ \Pirr S_1 ; (\Pirr f)_+\spbreve\] 
since for any $j_1\in J(S_1)$ and $m_2\in M(S_2)$, the condition $f(j_1)\not\leq_{S_2} m_2$ is equivalent to both
\[ \exists j_2\in J(S_2).[j_2\leq_{S_2} f(j_1)\wedge j_2\not\leq m_2] \quad\text{and}\quad \exists m_1\in M(S_1).[j_1\not\leq_{S_1} m_1 \wedge f_*(m_1)\leq m_2]
\]
using that $j_1$ is a join of elements in $J(S_1)$ and $m_2$ is a meet of elements of $M(J_2)$. Thus, $\Pirr(f)$ is a $\Dep$-morphism. Moreover,
\[\Pirr(\id_S) = \not\leq_S = \id_{\Pirr(S)}\quad\]
 and, for any $g\colon S_2\to S_3$,
\[\Pirr (g \circ f) =  \Pirr f; (\Pirr g)_+\spbreve = \Pirr f \fatsemi \Pirr g. \]
To see this, note that for $j_1\in J(S_1)$ and $m_3\in M(S_3)$,
\[ g(f(j_1)) \not\leq_{S_3} m_3 \Lra f(j_1)\not\leq_{S_2} g_*(m_3)\Lra \exists m_2\in M(S_2).[f(j_1) \not\leq_{S_2} m_2\wedge g_*(m_3)\leq_{S_2} m_2], \]
where the last equivalence uses that $M(S_2)$ meet-generates $S_2$.
%
%
\end{proof}

\begin{theorem}
The categories $\JSLf$ and $\Dep$ are equivalent via $\Pirr$ and $\Open$.  The natural isomorphisms $\rep\colon \Id_{\JSLf} \Rightarrow \Open\circ\Nleq$ and $\red\colon \Id_\Dep \Rightarrow \Nleq\circ\Open$ are given by 
\[
\begin{array}{llll}
\rep_S\colon S\xra{\cong} \Open(\Pirr(S)),& x\mapsto \{ m\in M(S): x \not\leq_S m \}&&(S\in \JSLf),\\
\red_\rR\colon \rR\xra{\cong} \Pirr(\Open(\rR)),& 
\red_\rR(x, Y) ~\mathrel{\vcentcolon\Longleftrightarrow}\ \rR[x] \nsubseteq Y && (\rR\in \Dep).
\end{array}
\]
\end{theorem}

\begin{proof}
\begin{enumerate}
\item
  For $\rep$ restrict the argument concerning $\alpha$ in the proof of \autoref{thm:nleq_open_equiv}. In particular, the map $\rep_S$ is injective because every element of $S$ is uniquely determined by the meet-irreducibles above it, and surjective because $\Pirr(S)[X]=\rep_S(\bigvee X)$ for $X\seq J(S)$.

\item  Likewise, for $\red$ restrict the argument concerning $\beta$. Additionally use that $J(\Open \rR) \subseteq \{ \rR[x] : x \in \rR_\src \}$ and $M(\Open\rR) \subseteq \{ \inte_\rR(\rR_\trg\setminus \{y\}) : y \in \rR_\trg \}$ by \autoref{lem:jirr-mirr-open}.\qedhere
\end{enumerate}
\end{proof}

\section*{Details for \autoref{rem:join_meet_generators}}
We prove the claim made in \autoref{rem:join_meet_generators} that for every $S\in \JSLf$ and sets $J,M\seq S$ of join- and meet-generators, the map
\[ h\colon \Open(\not\leq_S\cap J\times S)\to \Open(\not\leq_S\cap J(S)\times M(S)), \quad O\mapsto O\cap M(S),\]
defines an isomorphism of semilattices. For notational simplicity let us put 
\[\rR:=\not\leq_S\cap J(S)\times M(S)\qquad\text{and}\qquad \rS:=\not\leq_S\cap J\times S.\] 
First, $h$ is a well-defined map: for every $O\in \Open(\rS)$ we have $O=\rS[X]$ for some $X\seq J$, so
\[ h(O)=\rS[X]\cap M(S) = \rR[X\cap J(S)] \cap M(S) \in \Open(\rR), \]
using that for any $j\in J$ and $x\in S$ with $j\not\leq_S x$ one has $j'\not\leq_S x$ for some $j'\in J(S)$. Clearly $h$ preserves finite unions and is surjective. To show $h$ is injective, let $X,X'\seq J$ such that $\rS[X] \cap M(S)=\rS[X']\cap M(S)$. Thus, for all $m\in M(S)$ we have
\[ \forall x\in X: x\leq_S m\quad\iff\quad \forall x'\in X': x'\leq_S m.  \] 
Since $M(S)$ meet-generates $S$, this implies that the $\iff$ statement holds for all $m\in S$, in particular for $m\in M$; hence $\rS[X]=\rS[X']$.

\section*{Proof of \autoref{lem:dep-iso-preserves-dim}}
Part (1) was shown in the proof of \cite[Theorem 4.8]{mmu21}. Part (2) follows from (1): since $\Open\colon \Dep\to \JSLf$ is a functor, $\rR\cong \rS$ in $\Dep$ implies $\Open(\rR)\cong \Open(\rS)$ in $\JSLf$.

\section*{Proof of \autoref{lem:jirr-mirr-open}}
\begin{enumerate}
\item is obvious since the sets $\rR[x]$ ($x\in \rR_\src$) join-generate $\Open(\rR)$.
\item The sets $\intop_{\rR}(\rR_\trg\setminus \{y\})$ $(y\in Y$) meet-generate $\Open(\rR)$: for every $X\seq \rR_\src$ we have 
\[  \rR[X]=\bigcap_{y\not\in \rR[X]} \intop_{\rR}(\rR_\trg\setminus \{y\})\]
by definition of $\intop_{\rR}$. Then it is sufficient to prove that for all $y,y_1,y_2\in Y$, 
\[ \intop_{\rR}(\rR_\trg\setminus \{y\}) = \intop_{\rR}(\rR_\trg\setminus \{y_1\}) \wedge \intop_{\rR}(\rR_\trg\setminus \{y_2\})\quad\iff\quad \breve{\rR}[y]=\breve{\rR}[y_1]\cup \breve{\rR}[y_2]. \]
To this end, we compute
\begin{align*}
& \intop_{\rR}(\rR_\trg\setminus \{y\}) = \intop_{\rR}(\rR_\trg\setminus \{y_1\}) \wedge \intop_{\rR}(\rR_\trg\setminus \{y_2\})\\
\iff\quad & \intop_{\rR}(\rR_\trg\setminus \{y\}) = \intop_{\rR}(\rR_\trg\setminus \{y_1,y_2\})\\
\iff\quad & \forall x\in \rR_\src: \rR[x]\seq Y\setminus \{y\} \Lra \rR[x]\seq Y\setminus \{y_1,y_2\} \\
\iff\quad & \forall x\in \rR_\src: \rR(x,y) \Lra \rR(x,y_1) \vee \rR(x,y_2) \\
\iff\quad &  \breve{\rR}[y]=\breve{\rR}[y_1]\cup \breve{\rR}[y_2].\tag*{\qedhere}
\end{align*}
\end{enumerate}

\section{Details for \autoref{sec:nuclear-lattice}}
We first establish a number of equivalent descriptions of nuclear morphisms (\autoref{lem:nuclear_char}). For this purpose, we introduce the following special morphisms:

\begin{notation}
For any $S,T\in \JSLf$ and $s\in S$, $t\in T$, we define the semilattice morphism
\[ s\ostar_{S,T} t\colon S\to T,\qquad x\mapsto \begin{cases}
t, & x\not\leq_S s;\\
\bot_T, & x\leq_S s.
\end{cases}  \]
Whenever $S$ and $T$ are clear from the context, we drop subscripts and write $s\ostar t$ for $s\ostar_{S,T} t$.
\end{notation}

\begin{rem}\label{rem:up_mor_comp_closed}
  These morphisms are closed under composition:
  \[
    {t_2 \ostar u} \,\circ\, {s\ostar t_1}
    \;=\; \begin{cases}
      s\ostar u, & \text{$t_1 \nleq_T t_2$}
      \\
     \top_S, \ostar \bot_U & \text{$t_1 \leq_T t_2$}
    \end{cases}
    \qquad\text{for any}\quad
    s\ostar t_1 \colon S\to T\text{ and }t_2\ostar u\colon T\to U.
  \]
\end{rem}

\begin{rem}
  For any $S,T\in \JSLf$ the set $\JSL(S,T)$ of morphisms forms a finite semilattice where the join $f\vee g$ of $f,g\colon S\to T$ is given pointwise, i.e.\ $f\vee g(s)=f(s)\vee g(s)$ for $s\in S$.
\end{rem}

\begin{lemma}\label{lem:nuclear_char}
  \label{lem:nuclear_technical}
  For any $f\colon S \to T$ in $\JSLf$ the following are equivalent:
  \begin{enumerate}
	\item $f$ is nuclear.
    \item $f$ factorizes through a boolean algebra.
	\item $f$ is a join of morphisms $s\ostar t\colon S\to T$ where $s\in S$, $t\in T$.
    \item $f$ is a join of morphisms $m\ostar j\colon S \to T$ where $m \in M(S)$, $j \in J(T)$.
  \end{enumerate}
\end{lemma}

\begin{proof}
  \begin{enumerate}
    \item[(1)$\To$(2)]
    Assume $f$ factorizes through a finite distributive lattice $D$. It suffices to show $D$ arises as a retract $r \circ e = \id_D$ of a finite boolean algebra. Since join-irreducibles of finite distributive lattices are join-prime \cite{GratzerGeneralLattice1998}, we have the following well-defined retract, where $\Pow X$ denotes the $\cup$-semilattice of subsets of a set $X$.
    \[
      \begin{tabular}{ll}
        $e\colon D \monoto \Pow J(D)$,
        &  $e(d) := \{ j \in J(D) : j \leq_D d \}$, \\
	$r\colon \Pow J(D) \epito D$,
        & $r(S) := \Lor_D S$.
      \end{tabular}
    \]
    \item[(2)$\To$(1)] is trivial.

    \item[(2)$\Lra$(3)] 
      Suppose  $f\colon S \to T$ factorizes through a boolean algebra, i.e.\ $f = S \xto{g} \mathbb{2}^Z \xto{h} T$ for some finite set $Z$ and morphisms $g$, $h$. Since products and coproducts coincide in $\JSL_f$, the morphisms $g$ and $h$ can be decomposed as
\[ g=\langle g_z\rangle_{z\in Z} \qquad\text{and}\quad h=[h_z]_{z\in Z} \]
for some $g_z\colon S\to \mathbb{2}$ and $h_z\colon \mathbb{2}\to T$ ($z\in Z$). Note that $g_z=a_z\ostar 1$ where $a_z$ is the greatest element of $S$ with $g(a_z)=0$, and $h_z=0\ostar b_z$ for $b_z=h_z(1)$.
Thus:
      \begin{align*}
        f &= [h_z]_{z\in Z} \circ \langle g_z\rangle_{z\in Z} \\
       &= \Lor_{z\in Z}  h_z \circ g_z : 
        \\
        &= \Lor_{z\in Z} [0 \ostar b_z]\circ [a_z \ostar 1]
        \\ 
       & = \Lor_{z\in Z} a_z \ostar b_z.
      \end{align*}
Conversely, by reasoning backwards, if $f= \Lor_{z\in Z} a_z \ostar b_z$ for some $a_z\in S$, $b_z\in T$ ($z\in Z$) then $f$ factorizes through the boolean algebra $\mathbb{2}^Z$ via $\langle a_z\ostar 1 \rangle_{z\in I}$ and $[0\ostar b_z]_{z\in I}$.
\item[(4)$\To$(3)] is trivial.
 \item[(3)$\To$(4)] It suffices to show
    \[
      s \ostar t =
      \Lor \{ m \ostar j : m\in M(S),\; j\in J(T),\; s \leq_S m,\; j \leq_T t \}
      \qquad
      \text{for all $s\in S$, $t\in T$}.
    \]
    Clearly, each summand satisfies $m \ostar j \leq s \otimes t$, so ``$\geq$'' holds. Conversely, given $s'\not\leq_S s$ there exists $m\in M(S)$ with $s'\not\leq_S m$ and $s\leq_S m$, using that $s$ is a meet of meet-irredubibles. Thus, $m\ostar j$ is summand of the right-hand side for each $j\leq t$, and $[m\ostar j](s')=j$. Thus, the action of the right-hand side on $s'$ is no less than $\Lor_T \{ j \in J(T) : j \leq_T t \} = t$, proving ``$\leq$''.\qedhere
  \end{enumerate}
\end{proof}

\noindent The following lemma characterizes nuclear languages in terms of the relations $\rDR{L}, \rDR{L,a}\seq \LD{L}\times \LD{\rev{L}}$, cf. \autoref{ex:sld-vs-dlr}. It implies that nuclearity of $L$ is decidable in polynomial time when $\rDR{L}$ and $\rDR{L,a}$ are given.
\begin{lemma}\label{lem:recog_nuclear_langs} A regular language $L\seq \Sigma^*$ is nuclear iff for each $a \in \Sigma$, the relation $\rDR{L,a}$ is a union of sets of the form  $\rDR{L}\spbreve[v^{-1} L^r] \times \rDR{L}[u^{-1} L]$ where $u,v\in \Sigma^*$.
\end{lemma}

%

\begin{proof}
We simply translate nuclearity from $\JSLf$ into the equivalent category $\Dep$, recalling $J(\SLD{L})\seq \LD{L}$ and $M(\SLD{L})\seq \dr_L[\LD{\rev{L}}]$ by \eqref{eq:drl-iso}. By \autoref{lem:nuclear_char} we know $L$ is nuclear iff each $\delta_a\colon \SLD{L}\to \SLD{L}$ is a join of morphisms $\dr_L(v^{-1} \rev{L}) \ostar u^{-1} L$ for $u,v\in \Sigma^*$. Moreover $\delta_a\colon \SLD{L}\to \SLD{L}$ corresponds to the $\Dep$-morphism $\rDR{L,a}\colon \rDR{L}\to \rDR{L}$ by \autoref{ex:sld-vs-dlr}. Next, $\dr_L(v^{-1} L^r) \ostar u^{-1} L$ corresponds to the $\Dep$-morphism
\[
    \rDR{L}\spbreve[v^{-1} L^r] \times \rDR{L}[u^{-1} L] : \rDR{L} \to \rDR{L},
\]
as shown by the following computation for $x,y\in \Sigma^*$:
\begin{align*}
& \Pirr(\dr_L(v^{-1} L^r)\ostar u^{-1} L)(x^{-1}L,\dr_L(y^{-1}\rev{L})) & \\
\iff~~ & [\dr_L(v^{-1} L^r) \ostar u^{-1} L](x^{-1}L)\not\seq \dr_L(y^{-1}\rev{L}) & \text{def. $\Pirr$} \\
\iff~~ &x^{-1}L\not\seq \dr_L(v^{-1}\rev{L}) \text{ and } u^{-1}L\not\seq \dr_L(y^{-1}\rev{L}) & \text{def. $\ostar$} \\
 \iff~~ &\rDR{L}(x^{-1}L,v^{-1}\rev{L}) \text{ and } \rDR{L}(u^{-1}L,y^{-1}\rev{L}) & \text{by \eqref{eq_DR_vs_dr}} \\
\iff~~ & (x^{-1}L,y^{-1}\rev{L})\in \rDR{L}\spbreve[v^{-1} L^r] \times \rDR{L}[u^{-1} L].& 
\end{align*}
Finally observe $\Pirr$ translates joins of $\JSLf$-morphisms into unions of relations, i.e.\ we have $\Pirr(f\vee g)=\Pirr(f)\cup \Pirr(g)$ for any two parallel morphisms $f,g$.
%
\end{proof}

\section*{Proof of \autoref{lem:lattice_lang_nuclear}}
We first describe the minimal $\JSL$-dfa of a lattice language:
\begin{lemma}\label{lem:lattice_aut}
 For any $S\in \JSLf$, the minimal $\JSL$-dfa for $L(S)$ is given by 
\[ A_S=(S,\delta,\top_S, S\setminus \{\bot_S\})\qquad\text{where}\qquad \begin{tabular}{ll}
      $\delta_{\bra{j}} :=\ \bot_S\ostar j\colon S\to S$
      & for $j \in J(S)$,
      \\[1ex]
     $\delta_{\ket{m}} :=\ m\ostar \top_S\colon S\to S$
      & for $m \in M(S)$.
    \end{tabular} \] 
\end{lemma}

\begin{proof}
It is clear from the definition that $A_S$ accepts $L(S)$. The automaton $A_S$ is $\JSL$-reachable because every state $j\in J(S)$ is reached on input $\bra{j}$ from the initial state $\top_S$; thus, every state is a join of states reachable via transitions. To see that $A_S$ is simple, suppose that $s\neq s'$ are distinct states, w.l.o.g. $s\not\leq_S s'$. Then there exists $m\in M(S)$ such that $s\not\leq_S m$ and $s'\leq_S m$. Therefore, the state $s$ accepts $\ket{m}$ and the state $s'$ does not accept $\ket{m}$, showing that $L(A_S,s)\neq L(A_S,s')$. This proves that $A_S$ is a minimal $\JSL$-dfa for $L(S)$.
\end{proof}
This immediately implies \autoref{lem:lattice_lang_nuclear}: the isomorphism $S\cong \SLD{L(S)}$ follows from the uniqueness of minimal $\JSL$-dfas, and the nuclearity of $L$ by \autoref{lem:nuclear_technical}.

\section*{Proof of \autoref{prop:nuclear-complexity}(1)}

Since $\dim{\rDR{L}}\leq \ns{L}$ holds for all regular languages $L$ by \eqref{eq:complexity_ineq}, we need only prove $\ns{L}\leq \dim{\rDR{L}}$. Suppose that $\rDR{L}$ has a biclique cover of size $k$; our task is to construct an nfa for $L$ with at most $k$ states.
Since $\Pirr(\SLD{L})\cong \rDR{L}$ by \autoref{ex:sld-vs-dlr}, we have $\Open(\rDR{L})\cong \Open(\Pirr(\SLD{L}))\cong \SLD{L}$ in $\JSLf$. Thus, by \autoref{lem:dep-iso-preserves-dim}, there exists a $\JSLf$-monomorphism $e\colon \SLD{L}\monoto S$ into a finite semilattice $S$ with $|J(S)|\leq k$. We shall equip $S$ with the structure of a $\JSL$-dfa $A=(S,\gamma,i_S,f_S)$ such that $e$ 
is a $\JSL$-dfa morphism from the minimal $\JSL$-dfa $\SLD{L}$ into $A$, i.e. the following diagram commutes for all $a\in \Sigma$:
\[
\xymatrix@R-1.5em@C+1em{
& S \ar[r]^{\gamma_a} & S \ar[dr]^{f_S} & \\
\mathbb{2} \ar[ur]^{i_S} \ar[dr]_{i} & & & \mathbb{2} \\
& \SLD{L} \ar@{>->}[uu]^e \ar[r]_{\delta_a} & \SLD{L} \ar@{>->}[uu]_e \ar[ur]_{f} &
}
\]
The morphism $i_S$ is given by $i_S := e\circ i$. To define $f_S$, observe that by definition of the final states of $\SLD{L}$ we have $f= K \ostar 1$ where $K$ is the largest language in $\SLD{L}$ not containing $\epsilon$. Then $f_S:=e(K)\ostar 1$ satisfies $f_S\circ e=f$ as required. For the transitions $\gamma_a$ we use our assumption that $L$ is nuclear, i.e.\ 
\[\delta_a\;=\;\bigvee_{i=1}^p {x_i\ostar y_i}\qquad \text{for some $x_i,y_i\in \SLD{L}$,} \]
so we may define:
\[ \gamma_a \;:=\; \bigvee_{i=1}^p {e(x_i)\ostar e(y_i)}. \]
The central square commutes because composition of $\JSLf$-morphisms preserves joins in each argument and $[e(x_i) \ostar e(y_i)] \circ e = e \circ [x_i \ostar y_i]$, since $e(x_i) \nleq_S e(y_i) \iff x_i \nsubseteq y_i$ recalling that injective join-semilattice morphisms are order-embeddings.

Then we have proved $e$ to be a $\JSL$-dfa morphism. Since $\JSL$-dfa morphisms preserve the accepted language, we see that $A$ accepts the language $L$. Thus its corresponding nfa $J(A)$ of join-irreducibles is an nfa accepting $L$ with $|J(S)|\leq k$ states.

\section*{Proof of \autoref{prop:lattice-complexity}(2)}
 Let $L=L(S_0)$ be the lattice language for $S_0\in \JSLf$. Since $\dim{\rDR{L}}\leq \nalpha{L}$ holds for all regular languages $L$ by \eqref{eq:complexity_ineq}, we need only prove $\nalpha{L}\leq \dim{\rDR{L}}$. To this end, suppose $\rDR{L}$ has a biclique cover of size $k$; our task is to construct an atomic nfa for $L$ with at most $k$ states. Note that $S_0\cong \SLD{L}$ in $\JSLf$ by \autoref{lem:lattice_lang_nuclear}. Thus, by \autoref{lem:dep-iso-preserves-dim} there exists a $\JSLf$-monomorphism $e\colon \SLD{L}\monoto S$ into a finite semilattice $S$ with $|J(S)|\leq k$. We may assume $e(\top_{S_0})=\top_S$; otherwise, replace $S$ by the subsemilattice of all $s\in S$ with $s\leq_S e(\top_{S_0})$, which has no more join-irreducibles than $S$.

We equip $S$ with the structure of a $\JSL$-dfa $A=(S,\gamma,S\setminus \{\bot_{S}\},\top_S)$ where the transitions are given by $\gamma_{\bra{j}}=\bot_S \ostar e(j)$ for $j\in J(S_0)$ and $\gamma_{\ket{m}}=e(m)\ostar \top_{S}$ for $m\in M(S_0)$. This makes $e$ a $\JSL$-automata morphism from the minimal automaton $A_{S_0}$ for $L(S_0)$ (see \autoref{lem:lattice_lang_nuclear}) into the automaton $A$. Indeed, both triangles and the $j$- and $m$-square in the following diagram commute:
\[
\xymatrix@R-1.5em@C+1em{
& S \ar@<0.5ex>[r]^{\bot_S\ostar e(j)} \ar@<-0.5ex>[r]_{e(m)\ostar \top_S} & S \ar[dr]^{\bot_s\ostar 1} & \\
\mathbb{2} \ar[ur]^{\top_S} \ar[dr]_{\top_{S_0}} & & & \mathbb{2} \\
& S_0 \ar@{>->}[uu]^e \ar@<0.5ex>[r]^{\bot_{S_0}\ostar j} \ar@<-0.5ex>[r]_{m\ostar \top_{S_0}} & S_0 \ar@{>->}[uu]_e \ar[ur]_{\bot_{S_0}\ostar 1} &
}
\]
Note that for any starting state $s\in S$, the automaton $A$ reaches a state from $e[A_{S_0}]\cong A_{S_0}$ after reading the first input letter and the remaining computation takes place in that subautomaton.

Since $\JSL$-automata morphisms preserve the accepted language, we conclude that the automaton $A$ accepts the language $L$. We will show below that for every state $s\in S$ its accepted language $L(A,s)$ is invariant under the Nerode left congruence, i.e.\ for all $v,w\in \Sigma^*$,
\begin{equation}\label{eq:nerode} v\sim_L w \qquad\text{implies}\qquad v\in L(A,s) \iff w\in L(A,s). \end{equation}
This suffices to conclude the proof: since the congruence classes of $\sim_L$ are precisely the atoms of $\BLD{L}$, it follows from \eqref{eq:nerode} that every state of $A$ accepts some language from $\BLD{L}$, so the nfa $J(A)$ of join-irreducibles is an atomic nfa for $L$ with $|J(S)|\leq k$ states.

It remains to prove \eqref{eq:nerode}. Suppose $v\sim_L w$, i.e.\ $xv\in L$ iff $xw\in L$ for every $x\in \Sigma^*$. By definition of $L$, this means precisely that $v$ and $w$ match one of following three cases:

\medskip\noindent \textbf{Case 1.} $v=\bra{j}v'$ and $w=\bra{j'}w'$ for some $j,j'\in J(S_0)$ and $v',w'\in \Sigma^*$:
\[ v\in L(A,s)\iff v'\in L(A_{S_0},j) \iff \bra{j}v'\in L \iff \bra{j'}w'\in L \iff \cdots \iff w\in L(A,s). \]

\medskip\noindent \textbf{Case 2.} $v=\ket{m}v'$ and $w=\ket{m}w'$ for some $m\in M(S_0)$, and $v',w'\in \Sigma^*$. 

\medskip\noindent If $s\leq_{S_0} e(m)$ then $A$ goes to state $\bot_S$ after reading the first $\ket{m}$ and thus $v,w\not\in L(A,s)$. If $s\not\leq_{S_0} e(m)$, then
\[ v\in L(A,s)\Lra v'\in L(A_{S_0},\top_{S_0}) \Lra v'\in L \Lra \ket{m}v'\in L \Lra \ket{m}w'\in L \Lra \cdots \Lra  w\in L(A,s). \]

\medskip\noindent \textbf{Case 3.} $v=\epsilon$, $w\in L$ and $w=jw'$ for some $j\in J(S_0)$ and $w'\in \Sigma^*$ (or symmetrically).

\medskip\noindent The claim holds for $s=\bot_S$ since $L(A,\bot_S)=\emptyset$, so assume $s\neq \bot_S$. Then $v=\epsilon\in L(A,s)$. Moreover, $w=jw'\in L$ implies $w'\in L(A_{S_0},j)$, i.e. $w=jw'\in L(A,s)$. \qedhere

\section{Details for \autoref{sec:complexity}}

\section*{Proof of \autoref{thm:subatomic_npc}}

The proof is similar to the one of \autoref{thm:atomic_npc}. Again, we split it into two propositions:

\begin{proposition}\label{prop:subatomic_np}
The problem $\MONTOSUBATOMICNFA$ is in $\NP$.
\end{proposition}

\begin{proof} 
Let $(M,h,F)$ be a monoid recognizer for the language $L\seq\Sigma^*$, and let $k$ be a natural number. We claim the following three statements to be equivalent:
\begin{enumerate}[label=(\alph*)]
\item There exists a subatomic nfa accepting $L$ with at most $k$ states.
\item There exists a finite semilattice $S$ with $|J(S)|\leq k$ and $\JSLf$-morphisms $p$, $q$ and $\tau_a$ ($a\in \Sigma$) making the left diagram below commute.
\item There exists a $\Dep$-object $\rS\seq \rS_\src\times \rS_\trg$ with $|\rS_\src|\leq k$ and $|\rS_\trg|\leq |M|$ and $\Dep$-morphisms $\rP$, $\rQ$ and $\rT_a$ ($a\in \Sigma$) making the right diagram below commute (cf. \autoref{ex:sld-vs-dlr}/\ref{ex:blrd-vs-syn}).
\end{enumerate}
\[
\xymatrix@C-0.5em{
& \BLRD{L} \ar[r]^{\delta_a''} & \BLRD{L} \ar[dr]^{f''} & \\
\mathbb{2} \ar[ur]^{i''} \ar[dr]_{i} & S \ar@{-->}[u]^q \ar@{-->}[r]^{\tau_a} & S \ar@{-->}[u]_q & \mathbb{2} \\
& \SLD{L} \ar@{-->}[u]^p \ar@{-->}[r]_{\delta_a} & \SLD{L} \ar@{-->}[u]_p \ar[ur]_{f} & 
}
\qquad
\xymatrix@C-0.5em{
& \id_{\Syn{L}} \ar[r]^{\rD_a''} & \id_{\Syn{L}} \ar[dr]^{\rF''} & \\
\id_1 \ar[ur]^{\rI''} \ar[dr]_{\rI} & \rS \ar@{-->}[u]^{\rQ} \ar@{-->}[r]^{\rT_a} & \rS \ar@{-->}[u]_{\rQ} & \id_1 \\
& \rDR{L} \ar@{-->}[u]^\rP \ar[r]_{\rDR{L,a}} & \rDR{L} \ar@{-->}[u]_\rP \ar[ur]_{\rF} & 
}
\]
In fact, (a)$\Lra$(b) was shown in \autoref{thm:na_nmu_char}(2), and (b)$\Lra$(c) follows from the equivalence between $\JSLf$ and $\Dep$. To see this, note first that in the left diagram, by replacing $q$ with its image we may assume that $q$ is injective. By the self-duality of $\JSLf$, dualizing $q$ yields a surjective morphism from $\BLRD{L} \cong \BLRD{L}^\op$ to $S^\op$. Thus,
\[ |M(S)|=|J(S^\op)| \leq |J(\BLRD{L})| = |\Syn{L}| \leq |M|,\]
where the last step uses the minimality of the recognizer $(\Syn{L},\mu_L, F_L)$.

By \autoref{ex:sld-vs-dlr} and \ref{ex:blrd-vs-syn} the upper and lower path of the left diagram in $\JSLf$ correspond to the upper and lower path of the right-hand diagram in $\Dep$. Therefore, \autoref{thm:jsl_vs_dep} shows the two diagrams to be equivalent.

\medskip\noindent From (a)$\Lra$(c) we deduce that the finite relations $\rS$, $\rP$, $\rQ$ and $\rT_a$ constitute a short certificate for the existence of a subatomic nfa. Commutativity of the right-hand diagram can be checked in polynomial time w.r.t.\ the size of $M$ because the morphisms in the lower and upper path can be efficiently computed from the given recognizer $(M,h,F)$. Indeed, note that
 one can view the monoid $(M,\bullet, 1)$ as a dfa accepting $L$ with initial state $1_M$, final states $F$, and transitions given by $m\xra{a} m\bullet h(a)$ for $m\in M$ and $a\in \Sigma$. Similarly, the opposite monoid $M^\op=(M,\bullet^\op, 1)$ with multiplication $m\bullet^\op n := n\bullet m$ yields a dfa for $\rev{L}$. Thus, \autoref{rem:dfapair} shows that $\rDR{L}$ and the relations in the lower path of the diagram can be computed from $(M,h,F)$ in polynomial time. Moreover, by minimizing the dfa $M$ and computing its transition monoid, one obtains $\Syn{L}$ in polynomial time and can efficiently compute the relations in the upper path of the diagram.
\end{proof}

\begin{proposition}\label{prop:subatomic_nphard}
The problem $\MONTOSUBATOMICNFA$ is $\NP$-hard.
\end{proposition}

\begin{proof}
Similar to the proof of \autoref{prop:atomic_nphard}, we give a polynomial-time reduction from $\BICCOV$ to $\MONTOSUBATOMICNFA$. Given a pair $(\rR,k)$ of a finite relation $\rR$ and a natural number $k$, we again form the lattice language $L=L(S)$ where $S=\Open(\rR)$. The reduction is given by
\[ (\rR,k)\quad\longmapsto \quad ((\Syn{L},\mu_L, F_L),k). \]
From the proof of \autoref{prop:atomic_nphard} we know that this reduction is correct (i.e.\ $\dim{\rR}=\nmu{L}$) and that $\dfa{L}$ and $\dfa{\rev{L}}$ can be constructed in polynomial time. It remains to verify that also $\Syn{L}$, i.e.\ the transition monoid of $\dfa{L}$, can be constructed in polynomial time. Recall that the transitions of the minimal $\JSL$-dfa for $L$ (see \autoref{lem:lattice_aut}) are \[\delta_{\bra{j}}=\bot_S\ostar j\quad\text{for $j\in J(S)$}\qquad\text{and}\qquad \delta_{\ket{m}}=m\ostar \top_S\quad \text{for $m\in M(S)$}. \]
Therefore, \autoref{rem:up_mor_comp_closed} implies that its transition monoid consists precisely of the maps
  \[
    \begin{cases}
      m\ostar \top_S & (m\in M(S))
      \\
	  \bot_S\ostar j & (j\in J(S))	
	\\
      m\ostar j & (j\in J(S),\; m\in M(S))
      \\
	  \bot_S\ostar \top_S &
      \\	
	  \top_S \ostar \bot_S 
      \\
      \id_S
    \end{cases}
  \]
Since $\dfa{L}$ is the (dfa-)reachable part of the minimal $\JSL$-dfa, its transition monoid is given by the above maps restricted to $\dfa{L}$.
Thus, it has $O(|\LD{L}| \cdot |\LD{L^r}|)$ elements. In particular, $\Syn{L}$ can be efficiently computed as the transition monoid of $\dfa{L}$.\qedhere
\end{proof}

\section{Details for \autoref{sec:applications}}

\section*{Proof of \autoref{thm:nuclear_n_complete}}

\begin{enumerate}
\item The restriction of $\DFA + \rev{\DFA} \to \NFA$ to nuclear languages is in $\NP$: given a pair $A$ and $B$ of dfas accepting the languages $L$ and $\rev{L}$, one can check nuclearity of $L$ in polynomial time by computing $\rDR{L}$ and $\rDR{L,a}$ (see \autoref{rem:dfapair}) and verifying the conditions of \autoref{lem:nuclear_char}. Since  $\dim{\rDR{L}}=\ns{L}$ by \autoref{prop:nuclear-complexity}, a certificate for $\ns{L}\leq k$ is given by a biclique cover of $\rDR{L}$ with at most $k$ elements.
\item The $\NP$-hardness proof is identical to the one for $\DFA+\rev{\DFA}\to \NFA_{\mathbf{atm}}$ in \autoref{prop:atomic_nphard}: the lattice language $L$ used in the reduction $(\rR,k)\mapsto (\dfa{L},\dfa{\rev{L}},k)$ is nuclear by \autoref{lem:lattice_lang_nuclear} and it satisfies $\dim{\rR}=\ns{L}$ by \autoref{prop:lattice-complexity}. \qedhere
\end{enumerate}

\section*{Proof of \autoref{prop:group-language-complexity}}

For any regular language $L\seq \Sigma^*$, we have the surjective $\JSLf$-morphism
\[ \cl_L\colon \BLRD{L}\epito \BLD{L},\qquad K\mapsto \bigcap\{ K'\in \BLD{L} : K\seq K' \}. \]
In fact, $\cl_L$ is the dual of the inclusion map $\BLD{L}\monoto \BLRD{L}$, using that $\BLD{L}^\op\cong \BLD{L}$ and $\BLRD{L}^\op \cong \BLRD{L}$.  Now suppose that $L$ is a group language.
\begin{enumerate}
\item We claim that $\cl_L$ preserves transitions, i.e.\ the following diagram commutes for $a\in \Sigma$:
\begin{equation}\label{eq:prestrans}
\xymatrix{
\BLD{L} \ar[r]^{\delta_a'} & \BLD{L} \\
\BLRD{L} \ar@{->>}[u]^{\cl_L} \ar[r]_{\delta_a''} & \BLRD{L} \ar@{->>}[u]_{\cl_L}
}
\end{equation}
It suffices to prove that $\delta_a'\circ \cl_L$ and $\cl_L\circ \delta_a''$ agree on the atoms of $\BLRD{L}$, i.e.\
\begin{equation}\label{eq:prestransat} \delta_a'(\cl_L([w]_{\equiv_L})) = \cl_L(\delta_a''([w]_{\equiv_L}))\qquad \text{for all $w\in \Sigma^*$}.\end{equation}
Since $L$ is a group language, the transition maps $\delta_a\colon \LD{L}\to \LD{L}$ ($a\in \Sigma$) of its minimal dfa are bijective. This implies that $\delta_a'$ and $\delta_a''$ are semilattice isomorphisms. In particular, they map atoms to atoms, so
\[ 
\delta_a''([av]_{\equiv_L}) = a^{-1}[av]_{\equiv_L} = [v]_{\equiv_L} \quad\text{and}\quad \delta_a'([av]_{\sim_L}) = a^{-1}[av]_{\sim_L} = [v]_{\sim_L}\qquad\text{for $a\in \Sigma$, $v\in \Sigma^*$}.
\] 
Given $w\in \Sigma^*$, choose a natural number $n\geq 1$ such that $[a^nw]_{\equiv_L} = [w]_{\equiv_L}$, using again that $\Syn{L}$ is a group. Then also $[a^nw]_{\sim_L} = [w]_{\sim_L}$ because $\equiv_L\,\seq\, \sim_L$. Thus, \eqref{eq:prestransat} follows by
\begin{align*}
\delta_a'(\cl_L([w]_{\equiv_L})) &= \delta_a'([w]_{\sim_L}) \\
&= \delta_a'([a^nw]_{\sim_L}) \\
&= [a^{n-1}w]_{\sim_L} \\
&= \cl_L([a^{n-1}w]_{\equiv_L}) \\
&= \cl_L(\delta_a''([a^nw]_{\equiv_L})) \\
&= \cl_L(\delta_a''([w]_{\equiv_L})).
\end{align*}
\item Since $\nmu{L}\leq \nalpha{L}$ for all regular languages $L$, we only need to prove $\nalpha{L}\leq \nmu{L}$. Thus, suppose that $L$ has a subatomic nfa with $k$ states. By \autoref{thm:na_nmu_char}(2), this means that there exists a finite semilattice $S$ with $|J(S)|\leq k$ and morphisms $p$, $q$ and $\tau_a$ ($a\in \Sigma$) for which the left-hand diagram below commutes.
 \[
\xymatrix@C-0.5em{
& \BLRD{L} \ar[r]^{\delta_a''} & \BLRD{L} \ar[dr]^{f''} & \\
\mathbb{2} \ar[ur]^{i''} \ar[dr]_{i} & S \ar@{-->}[u]^q \ar@{-->}[r]^{\tau_a} & S \ar@{-->}[u]_q & \mathbb{2} \\
& \SLD{L} \ar@{-->}[u]^p \ar@{-->}[r]_{\delta_a} & \SLD{L} \ar@{-->}[u]_p \ar[ur]_{f} & 
}\qquad
\xymatrix@C-0.5em{
& \BLD{L} \ar[r]^{\delta_a''} & \BLD{L} \ar[dr]^{f''} & \\
\mathbb{2} \ar[ur]^{i''} \ar[dr]_{i} & S' \ar@{-->}[u]^{q'} \ar@{-->}[r]^{\tau_a'} & S' \ar@{-->}[u]_{q'} & \mathbb{2} \\
& \SLD{L} \ar@{-->}[u]^{e\circ p} \ar@{-->}[r]_{\delta_a} & \SLD{L} \ar@{-->}[u]_{e\circ p} \ar[ur]_{f} & 
}
\]
By \eqref{eq:prestrans}, the image factorization $\xymatrix{S \ar@{->>}[r]^<<<<<e & S' \ar@{>->}[r]^<<<<<{q'} & \BLD{L}}$ of the morphism $\cl_L\circ q$ yields a subsemilattice $q'\colon S'\monoto \BLD{L}$ closed under left derivatives. Note that $|J(S')|\leq |J(S)|$ because $e\colon S\epito S'$ is surjective. Putting $\tau_a'=a^{-1}(\dash)\colon S'\to S'$ ($a\in \Sigma$) we see that the two squares in the right-hand diagram commute. Since $\cl_L(K)=K$ for $K\in \SLD{L}$ and $q\circ p$ is the inclusion from $\SLD{L}$ into $\BLRD{L}$, the two other parts of the diagram also commute. Thus, \autoref{thm:na_nmu_char}(1) shows that $L$ has an atomic nfa with $|J(S')|\leq k$ states. \qedhere
\end{enumerate}

\end{document}